\documentclass[11pt,authoryear]{elsarticle}
\usepackage{mydef}
\usepackage[boxed]{algorithm2e}
\usepackage{rotating,booktabs}
\usepackage{multirow}
\usepackage{mathtools}
\usepackage{tikz}
\usepackage[capitalise]{cleveref}
\usetikzlibrary{calc,decorations.markings}
\usetikzlibrary{arrows,patterns,positioning}
\usetikzlibrary{shapes.geometric, backgrounds}

\allowdisplaybreaks[3]

\def\calT{{\cal T}}
\def\calD{{\nabla}}
\def\disD{{\cal D}}
\def\disQ{{\cal Q}}
\def\calA{{\cal A}}

\def\dgxt{{D_{\Gamma_{X(t)}}}}
\def\gxt{{\Gamma_{X(t)}}}

\def\bs{\hfill $\blacksquare$}

\allowdisplaybreaks[3]

\newcommand{\irg}[3]{\ #1 \in \mathbb{Z} \cap [#2,#3]}

\newcommand{\CO}[2]{C^{(1)}_{#2,#1}}
\newcommand{\CT}[2]{C^{(2)}_{#2,#1}}
\newcommand\pFq[3]{{}_{#1}F_{#2}\left(#3\right){}}

\DeclarePairedDelimiterX\MeijerM[3]{\lparen}{\rparen}%
{\begin{smallmatrix}#1 \\ #2\end{smallmatrix}\delimsize\vert\,#3}

\SetAlCapSkip{1em}
\graphicspath{{./Graphs/}}

\begin{document}

\begin{frontmatter}

\title{An Expanded Local Variance Gamma model}

\author{P.~Carr}
\ead{petercarr@nyu.edu}

\author{A.~Itkin\corref{cor1}}
\ead{aitkin@nyu.edu}

\address{Tandon School of Engineering, New York University, \\
12 Metro Tech Center, RH 517E, Brooklyn NY 11201, USA}

\begin{abstract}
The paper proposes an expanded version of the Local Variance Gamma model of Carr and Nadtochiy by adding drift to the governing underlying process. Still in this new model it is possible to derive an ordinary differential equation for the option price which plays a role of Dupire's equation for the standard local volatility model. It is shown how calibration of multiple smiles (the whole local volatility surface) can be done in such a case. Further, assuming the local variance to be a piecewise linear function of strike and piecewise constant function of time this ODE is solved in closed form in terms of  Confluent hypergeometric functions. Calibration of the model to market smiles does not require solving any optimization problem and, in contrast, can be done term-by-term by solving a system of non-linear algebraic equations for each maturity, which is faster.
\end{abstract}


\end{frontmatter}

\section{Introduction}

Local volatility model was introduced by \cite{Dupire:94} and \cite{derman/kani:94} as a natural extension of the celebrating Black-Scholes model to take into account an existence of option smile. It is able to exactly replicate the local volatility function $\sigma(T,K)$ where $K,T$ are the option strike and time to maturity, at given pairs $(T,K)$  where the European options prices or their implied volatilities are known. This process is called calibration of the local volatility (or, alternatively, implied volatility) surface. Various approaches to solving this important problem were proposed, see, eg, survey in  \cite{ItkinLipton2017} and references therein.

As mentioned in \cite{ItkinLipton2017}, there are two main approaches to solving the calibration problem. The first approach relies on some parametric or non-parametric regression to construct a continuous implied volatility (IV) surface matching the given market quotes. Then the corresponding local volatility surface can be found via the well-known Dupire's formula, see, e.g., \cite{ItkinSigmoid2015} and references therein.

The second approach relies on the direct solution of the Dupire equation using either analytical or numerical methods. The advantage of this approach is that it guarantees no-arbitrage \footnote{But only if an analytical or numerical method in use does preserve no-arbitrage. This includes various interpolations, etc.}. However, the problem of the direct solution can be ill-posed, \cite{Coleman2001}, and is rather computationally expensive. For instance, in \cite{ItkinLipton2017} the Dupire equation (a partial differential equation (PDE) of the parabolic type) is solved by i) first using the Laplace-Carson transform, and ii) then applying various transformations to obtain a closed form solution of the transformed equation in terms of Kummer hypergeometric functions. Still, it requires an inverse Laplace transform to obtain the final solution.

With the second approach in use one also has to make an assumption about the behavior of the local/implied volatility surface at strikes and maturities where the market quotes are not known. Usually, by a tractability argument the corresponding local variance is seen either piecewise constant, \cite{LiptonSepp2011}, or piecewise linear \cite{ItkinLipton2017} in the log-strike space, and piecewise constant in the time to maturity space \footnote{See, however, comments in \cite{ItkinLipton2017} about their assumptions.}.

To improve computational efficiency of calibration, an important step is made in \cite{CarrNadtochiy2017} where Local Variance Gamma (LVG) model has been introduced (the first version refers to 2014 and can be found in \cite{CarrLVGOrig2014}).  This model assumes that the risk-neutral process for the underlying futures price is a pure jump Markov martingale, and that European option prices are given at a continuum of strikes and at one or more maturities. The authors construct a time-homogeneous process which meets a single smile and a piecewise time-homogeneous process, which can meet multiple smiles. However, in contrast to eg, \cite{ItkinLipton2017}, their construction leads not to a PDE, but to a partial differential difference equation (PDDE), which permits both explicit calibration and fast numerical valuation. In particular, it does not require application of any optimization methods, rather just a root solver. In \cite{CarrNadtochiy2017} this model is used to calibrate the local volatility surface assuming its piecewise constant structure in the strike space.

One of the potential criticism of this calibration method is the fact that the resulting local volatility function has a finite number of discontinuities. So it would be advantaged to relax the piecewise constant behavior of the surface. This is similar to how \cite{ItkinLipton2017} was developed to overcome the same problem as compared with \cite{LiptonSepp2011}.

On this way, recently \cite{FalckDeryabin2017} applied the LVG model to the FX options market where usually option prices are quoted only at five strikes. They assumed that the local volatility function is continuous, piecewise linear in the four inner strike subintervals and constant in the outer subintervals. A closed form solution of the PDDE derived in \cite{CarrLVGOrig2014}) is obtained with this parametrization, and calibration of some volatility smiles is provided. Still, to calibrate the model the authors rely on a residual minimization by using a least-square approach. So, despite an improved version of the LVG model is used, computational efficiency of this method is not perfect.

Another remark of \cite{CarrNadtochiy2017} is about the limitation that the risk-neutral price process of the underlying is assumed to be a martingale, i.e. the main driving process in \eqref{D} doesn't have a drift. However, the drift may not be negligible. If the drift is deterministic, e.g when the interest rate and dividends are deterministic, and the drift is a deterministic function of them,  the calibration problem can be reduced to the driftless case by discounting, but this assumption might be inconsistent with the market. Therefore, an expansion of the proposed model that allows for a non-zero and stochastic drift is very desirable. In particular, it would be interesting to expand the LVG model to a risk-neutral price process obtained by stochastic time change of a drifted diffusion. In this way, similar to local Variance Gamma model, \cite{MadCarrChang}, we introduce both stochastic volatility and stochastic drift.

With this in mind, our ultimate goals in this paper are as follows. First, we propose an expanded version of the LVG model by adding drift to the governing underlying process. It turns out that to proceed we need to re-derive and re-think every step in construction proposed in \cite{CarrNadtochiy2017}. We show that still it is possible to find an ordinary differential equation (ODE) for the option price which plays a role of Dupire's equation for the standard local volatility model, and how calibration of multiple smiles (the whole local volatility surface) can be done in such a case.

Further, assuming the local variance to be a piecewise linear function of strike and piecewise constant function of time we solve this ODE in closed form in terms of Confluent hypergeometric functions. Calibration of the model to market smiles does not require solving any optimization problem. In contrast, it can be done term-by-term by solving a system of non-linear algebraic equations for each maturity, and thus is much faster.

The rest of the paper is organized as follows. In Section~\ref{Process} the Expanded Local Variance Gamma model is formulated. In Section~\ref{ForwardEq} we derive a forward equation (which is an ordinary differential equation (ODE)) for Put option prices using a homogeneous Bochner subordination approach. Section~\ref{LVpwconst} generalizes this approach by considering the local variance being piece-wise constant in time. In Section~\ref{SolutionODE} a closed form solution of the derived ODE is given in terms of Confluent hypergeometric functions. The next Section discusses computation of a source term of this ODE which requires a no-arbitrage interpolation. Using the idea of \cite{ItkinLipton2017}), we show how to construct non-linear interpolation which provides both no-arbitrage, and a nice tractable representation of the source term, so that all integrals in the source term can be computed in closed form. In Section~\ref{calib} calibration of multiple smiles in our model is discussed in detail. To calibrate a single smile we derive a system of nonlinear algebraic equations for the model parameters, and explain how to obtain a smart guess for their initial values.  In Section~\ref{asympt}
asymptotic solutions of our ODE at extreme values of the model parameters are derived which improve computational accuracy and speed of the numerical solution. Section~\ref{numExp} presents the results of some numerical experiments where calibration of the model to the given market smiles is done term-by-term. The last Section concludes.

\section{Process} \label{Process}

Below where possible we follow the notation of \cite{CarrNadtochiy2017}.

Let $W_t$ be a $\mathbb{Q}$ standard Brownian motion with time index $t \ge 0$. Consider a stochastic process $D_t$ to be a time-homogeneous diffusion
\begin{equation} \label{D}
 d D_t = \mu D_t d t + \sigma(D_t) d W_t,
\end{equation}
\noindent where the volatility function $\sigma$ is local and time-homogeneous, and  $\mu$ is deterministic.

A unique solution to \eqref{D} exists if $\sigma(D) : \mathbb{R} \to \mathbb{R}$ is
Lipschitz continuous in $D$ and satisfies growth conditions at infinity. According to \eqref{D} we have $D_t  \in (-\infty,\infty)$ while $t \in [0,\infty)$. Since $D$ is a time-homogeneous Markov process, its infinitesimal generator $\calA$ is given by
\begin{equation} \label{gen}
\calA \phi(D) \equiv \left[\mu D \calD_D  + \frac{1}{2}\sigma^2(D) \calD^2_D \right] \phi(D)
\end{equation}
\noindent for all twice differentiable functions $\phi$. Here $\calD_x$ is a first order
differential operator on $x$. The semigroup of the $D$ process is
\begin{equation} \label{semig}
\calT^D_t \phi(D_t) = e^{t \calA} \phi(D_t)  = \EQ[\phi(D_t)|D_0 = D], \quad \forall t \ge 0.
\end{equation}

In the spirit of Variance Gamma model, \cite{MadanSeneta:90,MadCarrChang} and similar to \cite{CarrNadtochiy2017}, introduce a new process $D_{\Gamma_t}$ which is $D_t$ subordinated by the unbiased Gamma clock $\Gamma_t$. The density of the unbiased Gamma clock $\Gamma_t$ at time $t \ge 0$ is
\begin{equation} \label{gammaDen}
\mathbb{Q}\{\Gamma_t \in d\nu\} = \dfrac{\nu^{m-1} e^{-\nu m /t}}{(t^*)^m \Gamma(m)} d\nu, \quad \nu > 0, \quad m \equiv t/t^*.
\end{equation}
Here $t^* > 0$ is a free parameter of the process, $\Gamma(x)$ is the Gamma function. It is easy to check that
\begin{equation} \label{gammaExp}
\EQ[\Gamma_t] = t.
\end{equation}
\noindent Thus, on average the stochastic gamma clock $\Gamma_t$ runs synchronously with the calendar time $t$.

As applied to the option pricing problem, we introduce a more complex construction.
Namely, consider options written on the underlying process $S_t$. Without loss of generality and for the sake of clearness let us treat below $S_t$ as the stock price process. Here, in contrast to \cite{CarrNadtochiy2017}, we don't ignore interest rates $r$ and continuous dividends $q$ assuming them to be deterministic (below for simplicity of presentation we treat them as constants, but this can be easily relaxed). Then, let us define $S_t$ as
\begin{equation} \label{sub}
S_t = D_{\Gamma_{X(t)}}, \qquad X(t) = \dfrac{1 - e^{-(r-q) t}}{r-q}.
\end{equation}
\noindent It is clear that in the limit $ r \to 0, \ q \to 0$ we have $X(t) = t$, i.e., in this limit our construction coincides with that in \cite{CarrNadtochiy2017} who considered a driftless diffusion and assumed $S_t = D_{\Gamma_t}$. Also based on \eqref{gammaExp}
\begin{equation} \label{gammaXexp}
\EQ[\Gamma_{X(t)}] = X(t).
\end{equation}
Function $X(t)$ starts at zero, ie, $X(0) = 0$, and is a continuous increasing function of time $t$. Indeed, if $r - q > 0$, then $X(t)$ is increasing in $t$ on $t \in [0,\infty)$, and at $t \to \infty$ it tends to constant. The infinite time horizon is not practically important, but for any finite time $t$ function $X(t)$ can be treated as an increasing function in $t$. If $r - q < 0$, function $X(t)$ is strictly increasing $\forall t \in [0,\infty)$. Thus, $X(t)$ has all properties of a good clock. Accordingly, $\Gamma_{X(t)}$ has all properties of a random time.

Under a risk-neutral measure $\mathbb{Q}$, the total gain process, including the underlying price appreciation and dividends, after discounting at the risk free rate should be a martingale, see, eg, \cite{Shreve:1992}. This process obeys the following stochastic differential equation
\begin{align} \label{mart}
d \left( e^{-r t} S_t e^{q t} \right) &=  e^{(q-r)t} \left[ (q-r) S_t dt + d S_t\right]. \end{align}
Taking an expectation of both parts we obtain
\begin{align} \label{martCond}
\EQ[d \left( e^{(q-r)t} S_t\right)] &=
e^{(q-r)t} \left\{ (q-r)  \EQ[S_t] d t + d \EQ[S_t] \right\}.
\end{align}
Observe, that from \eqref{sub}, \eqref{D}
\begin{align} \label{e1}
\EQ[d S_t] &= \EQ[d \dgxt] = \mu \EQ [\dgxt d \gxt] + \EQ[\sigma(D_{\Gamma_{X(t)}}) d W_{\Gamma_{X(t)}}] = \mu \EQ [\dgxt d \gxt],
\end{align}
\noindent because the process $W_{\Gamma_t}$ is a local martingale, see \cite{RevuzYor1999}, chapter 6. Accordingly, the process $W_{\Gamma_{X(t)}}$ inherits this property from $W_{\Gamma_t}$, hence $\EQ[\sigma(D_{\Gamma_{X(t)}}) d W_{\Gamma_{X(t)}}] = 0$.

Further assume that the Gamma process $\Gamma_t$ is independent of $W_t$ (and, accordingly, $\gxt$ is independent of $W_\gxt$). Then the expectation in the RHS of \eqref{e1} can be computed, by first conditioning on $\gxt$, and then integrating over the distribution of $\gxt$ which can be obtained from \eqref{gammaDen} by replacing $t$ with $X(t)$, i.e.
\begin{align} \label{e2}
\EQ [\dgxt d \gxt  | S_s] &= \int_0^\infty  \EQ [\dgxt d \gxt| \gxt = \nu] \dfrac{\nu^{m-1} e^{-\nu m /X(t)}}{(t^*)^m \Gamma(m)} \\
&= \int_0^\infty  \EQ [D_\nu] \dfrac{\nu^{m-1} e^{-\nu m /X(t)}}{(t^*)^m \Gamma(m)} d \nu, \quad
\nu > 0, \quad m \equiv X(t)/t^*. \nonumber
\end{align}
The find $\EQ [D_\nu]$ we take into account \eqref{D} to obtain
\begin{align} \label{e3}
d \EQ [D_\nu] =  \EQ [d D_\nu]
= \EQ [\mu D_\nu d\nu + \sigma(D_\nu) D_\nu d W_\nu] = \mu \EQ [D_\nu] d\nu.
\end{align}
Solving this equation with respect to $y(\nu) = \EQ [D_\nu  | D_s]$,  we obtain $\EQ [D_\nu  | D_s] = D_s e^{\mu (\nu-s)}$. Since we condition on time $s$, it means that $D_s = D_{\Gamma_{X(s)}} = S_s$, and thus
$\EQ [D_\nu  | D_s] = S_s e^{\mu (\nu-s)}$.

Further, we substitute this into \eqref{e2}, set the parameter of the Gamma distribution $t^*$ to be $t^* = X(t)$ (so $m = 1$) and integrate to obtain
\begin{equation} \label{intRight}
d \EQ[S_t  | S_s] = \EQ[d S_t  | S_s] = \mu \EQ [\dgxt d \gxt] = S_s e^{- s \mu} \frac{\mu}{1 - \mu X(t)}.
\end{equation}
Setting now $m = r-q$ and solving this equation we find
\begin{equation} \label{ESsol}
\EQ[S_t  | S_s] = S_s (r-q) e^{(q-r)(s-t)}.
\end{equation}
Substituting \cref{ESsol} and \cref{intRight} into \cref{martCond} yields $d \left( e^{-r t} S_t e^{q t} \right) = 0$. Thus, if we chose $\mu = r-q$, the right hands part of \eqref{mart} vanishes, and our discounted stock process with allowance for non-zero interest rates and continuous dividends becomes a martingale. So the proposed construction can be used for option pricing.

This setting can be easily generalized for time-dependent interest rates $r(t)$ and continuous dividends $q(t)$. We leave it for the reader.

The next step is to consider connection between the original and time-changed processes. It is known from \cite{BochnerPDE1949}  that the process $G_{\Gamma_t}$ defined as
\[ d G_t = \sigma^2(G) d W_t \]
\noindent is a time-homogeneous Markov process. As the deterministic process $\mu t$ is also time-homogeneous, the whole process $D_t$ defined in \eqref{D} is also a time-homogeneous Markov process. Accordingly, the semigroups $T^S_t$ of $S_t$ and $T^D_t$ of $D_{\Gamma_{X(t)}}$ are connected by the Bochner integral
\begin{equation} \label{BI}
\calT^S_t U(S) = \int_0^\infty \calT^D_\nu U(S) \mathbb{Q}\{\gxt \in d\nu\}, \quad \forall t \ge 0,
\end{equation}
\noindent where $U(S)$ is a function in the domain of both $\calT^D_t$ and $\calT^S_t$.
It can be derived by exploiting the time homogeneity of the $D$ process, conditioning on the gamma time first, and taking into account the independence of $\Gamma_t$ and $W_t$ (or $\Gamma_\gxt$ and $W_\gxt$ in our case).

We set parameter $t^*$ of the gamma clock to $t^* = X(t)$. Then \eqref{BI} and \eqref{gammaDen} imply
\begin{equation} \label{BI1}
\calT^S_{t} U(S) = \int_0^\infty \calT^D_\nu U(S) \dfrac{e^{-\nu/X(t)}}{X(t)} d\nu.
\end{equation}
In what follows for the sake of brevity we will call this model as an Expanded Local Variance Gamma model, or ELVG.

\section{Forward equation for Put option prices} \label{ForwardEq}

Following \cite{CarrNadtochiy2017} we interpret the index $t$ of the semigroup $\calT^S_t$ as the maturity date $T$ of a European claim with the valuation time $t = 0$. Also let the test function $U(S)$ be the payoff of this European claim, ie,
\begin{equation} \label{payoff}
U(S_T) = e^{-r T}(K - S_T)^+.
\end{equation}
Then define
\begin{equation} \label{P0}
P(S_0,T,K) = \calT^S_T U(S_0)
\end{equation}
\noindent as the European Put value with maturity $T$ at time $t=0$ in the ELVG model. Similarly
\begin{equation} \label{P0D}
P^D(S_0,\nu,K) = \calT^D_\nu U(S_0)
\end{equation}
\noindent would be the European Put value with maturity $\nu$ at time $t=0$ in the model of \eqref{D}\footnote{Below for simplicity of notation we drop the subscript '0'  in $S_0$.}. Then the Bochner integral in \eqref{BI1} takes the form
\begin{equation} \label{Bochner2}
P(S,T,K) =  \int_0^\infty P^D(S,\nu,K) p e^{- p \nu} d \nu,  \quad p \equiv 1/X(T).
\end{equation}
Thus, $P(S,X(T),K)$ is represented by a Laplace-Carson transform of $P^D(S,\nu,K)$ with $p$ being a parameter of the transform. Note that
\begin{equation} \label{init}
P(S,0,K) = P^D(S,0,K) = U(S).
\end{equation}
To proceed, we need an analog of the Dupire forward PDE for $P^D(S,\nu,K)$.

\subsection{Derivation of the Dupire forward PDE \label{dupFWPDE}}

Despite this can be done in many different ways, below for the sake of compatibility we do it in the spirit of \cite{CarrNadtochiy2017}.

First, differentiating \eqref{P0D} by $\nu$ with allowance for \eqref{semig} yields
\begin{align} \label{Dt}
\calD_\nu P^D(S,\nu,K) &= e^{-r \nu} e^{\nu \calA}\left[ \calA - r\right] U(S)
= e^{-r \nu} \EQ \left[ \calA - r\right] U(S).
\end{align}
We take into account the definition of the generator $\calA$ in \eqref{gen}, and also remind that at $t=0$ we have $D_0 = S_0$. Then \eqref{Dt} transforms to
\begin{equation} \label{Dt1}
\calD_\nu P^D(S,\nu,K) = -r P^D(S,\nu,K) + (r-q) S \calD_S P^D(S,\nu,K) + e^{-r \nu}\frac{1}{2} \EQ \left[ \sigma^2(S) \calD_S^2 U(S) \right].
\end{equation}
However, we need to express the forward equation using a pair of independent variables $(\nu,K)$ while \eqref{Dt} is derived in terms of $(\nu,S)$. To do this, observe that
\begin{align} \label{sir}
e^{-r \nu} \EQ \left[\sigma^2(S) \calD_S^2 U(S) \right] &= e^{-r \nu}\EQ \left[\sigma^2(S) \delta(K-S)\right] = e^{-r \nu} \EQ \left[\sigma^2(K) \delta(K-S)\right] \\
&=  e^{-r \nu} \EQ \left[\sigma^2(K) \calD_K^2 U(S)\right] = \sigma^2(K) \calD_K^2 P^D(S,\nu,K). \nonumber
\end{align}
\noindent where the sifting property of the Dirac delta function $\delta(S-K)$ has been used. Also
\begin{align} \label{term2}
-r & P^D(S,\nu,K) + (r-q) S \nabla_S P^D(S,\nu,K) \\
&= e^{-r \nu} \EQ\left[ -r (K-S)^+ + (r-q)S \fp{(K-S)^+}{S} \right] \nonumber \\
&= e^{-r \nu} \EQ\left[ -r (K-S)^+ - (r-q)(K-S) \fp{(K-S)^+}{S} + (r-q)K \fp{(K-S)^+}{S}\right] \nonumber \\
&= e^{-r \nu} \EQ\left[ -r (K-S)^+ + (r-q)(K-S)^+  - (r-q)K \fp{(K-S)^+}{K}\right] \nonumber \\
&= -q P^D(S,\nu,K) - (r-q) K \nabla_K P^D(S,\nu,K). \nonumber
\end{align}
Therefore, using \eqref{sir} and \eqref{term2}, \eqref{Dt} could be transformed to
\begin{align} \label{Dup1}
\nabla_\nu P^D(S,\nu,K) &= -q P^D(S,\nu,K) - (r-q) K \nabla_K P^D(S,\nu,K) + \frac{1}{2} \sigma^2(K) K^2 \nabla_K^2 P^D(S,\nu,K) \nonumber  \\
&\equiv {\calA}^K P^D(S,\nu,K),  \\
{\calA}^K &= -q  - (r-q) K \nabla_K  + \frac{1}{2} \sigma^2(K) K^2 \nabla_K^2. \nonumber
\end{align}
This equation looks exactly like the Dupire equation with non-zero interest rates and continuous dividends, see, eg, \cite{Tysk2012} and references therein. Note, that $\calA^K$ is also a time-homogeneous generator.

\subsection{Forward partial divided-difference equation} \label{FPDDder}

Our final step is to apply the linear differential operator ${\calA}^K$ defined in  \eqref{Dup1} to both parts of \eqref{Bochner2}. Using time-homogeneity of $D_t$ and, again, the Dupire equation \eqref{Dup1}, we obtain
\begin{align} \label{b2}
-q &P(S,T,K) - (r-q) K \calD_K P(S,T,K) + \frac{1}{2} \sigma^2(K) \calD_K^2
P(S,T,K) \\
&= \int_0^\infty p e^{- p \nu} \left[-q P^D(S,\nu,K) - (r-q) K \calD_K P^D(S,\nu,K) + \frac{1}{2} \sigma^2(K) \calD_K^2 P^D(S,\nu,K)\right] d \nu \nonumber \\
&= \int_0^\infty p e^{- p \nu} \calD_\nu P^D(S,\nu,K) d \nu  = - p P^D(S,0,K) + p \int_0^\infty P^D(S,\nu,K) p e^{- p \nu} d \nu \nonumber \\
&= p\left[P(S,T,K) - P^D(S,0,K) \right] = p\left[P(S,T,K) - P(S,0,K) \right], \nonumber
\end{align}
\noindent where in the last line \eqref{init} was taken into account.

Thus, finally $P(S,T,K)$ solves the following problem
\begin{align} \label{finDup}
-q P(S,T,K) &- (r-q) K \calD_K P(S,T,K) + \frac{1}{2} \sigma^2(K) \calD_K^2
P(S,T,K) \\
&= \dfrac{P(S,T,K) - P(S,0,K)}{X(T)}, \qquad P(S,0,K) = (K-S)^+. \nonumber
\end{align}
At $ r = q = 0$ this equation translates to the corresponding equation in \cite{CarrNadtochiy2017}. In contrast to the Dupire equation which belongs to the class of PDE, \eqref{finDup} is an ODE, or, more precisely, a partial divided-difference equation (PDDE), since the derivative in time in the right hands part is now replaced by a divided difference. In the form of an ODE it reads
\begin{equation} \label{finDupPut}
\left[\frac{1}{2} \sigma^2(K) \calD_K^2 - (r-q) K \calD_K - \left(q + \dfrac{1}{X(T)}\right) \right] P(S,T,K) =
- \dfrac{P(S,0,K)}{X(T)}.
\end{equation}
This equation could be solved analytically for some particular forms of the local volatility function $\sigma(K)$ which are considered later in this paper. Also in the same way a similar equation could be derived for the Call option price $C_0(S,T,K)$ which reads
\begin{align} \label{finDupCall}
\Big[\frac{1}{2} \sigma^2(K) \calD_K^2 + (r-q) K \calD_K &- \left(q + \dfrac{1}{X(T)}\right) \Big] C_0(S,T,K) = - \dfrac{C_0(S,0,K)}{X(T)}, \nonumber \\
C_0(S,0,K) &= (S-K)^+.
\end{align}

Solving \eqref{finDupPut} or \eqref{finDupCall} provides the way to determine $\sigma(K)$ given market quotes of Call and Put options with maturity $T$. However, this allows calibration of just a single term. Calibration of the whole local volatility surface, in principle, could be done term-by-term (because of the time-homogeneity assumption) if \eqref{finDupPut}, \eqref{finDupCall} could be generalized to this case. We consider this in the following Section.

\section{Local variance piece-wise constant in time} \label{LVpwconst}

To address calibration of multiple smiles, we need to relax some assumptions about time-homogeneity of the process $D_t$ defined in \eqref{D}. This includes several steps which are described below in more detail.

\subsection{Local variance}

Here we assume that the local variance $\sigma(D_t)$ is no more time-homogeneous, but a piece-wise constant function of time $\sigma(D_t,t)$.

Let $T_1, T_2, \ldots, T_M$ be the time points at which the variance rate $\sigma^2(D_t)$ jumps deterministically. In other words, at the interval $t \in [T_0,T_1)$, the variance rate is $\sigma^2_0(D_t)$, at $t \in [T_1, T_2)$ it is $\sigma^2_1(D_t)$, etc. This can be also represented as
\begin{align} \label{sigmaPW}
\sigma^2(D_t,t) &= \sum_{i=0}^M \sigma^2_i(D_t) w_i(t), \\
w_i(t) &\equiv {\mathbf 1}_{t - T_i} - {\mathbf 1}_{t - T_{i+1}}, \ i=0,\ldots,M,
\quad T_0 = 0, \ T_{M+1} = \infty, \quad {\mathbf 1}_x =
\begin{cases}
1, & x \ge 0 \\
0, & x < 0.
\end{cases}
\nonumber
\end{align}

Note, that
\[ \sum_{i=0}^M w_i(t) = {\mathbf 1}_t - {\mathbf 1}_{t-\infty} = 1, \quad \forall t \ge 0.\]
\noindent Therefore, in case when all $\sigma^2_i(D_t)$ are equal, ie, independent on index $i$, \eqref{sigmaPW} reduces to the case considered in the previous Sections.

It is important to notice that our construction implies that the volatility $\sigma(D_t)$ jumps as a function of time at the calendar times $T_0, T_1,\ldots,T_M$, and not at the business times $\nu$ determined by the gamma clock. Otherwise, the volatility function would change at random (business) times which means it is stochastic. But this definitely lies out of scope of our model. Therefore, we need to change \eqref{sigmaPW} to
\begin{align} \label{sigmaPW_exp}
\sigma^2(\dgxt, \gxt) &= \sum_{i=0}^M \sigma^2_i(D_t) {\bar w}_i(\EQ(\gxt)), \\
{\bar w}_i(t) &= {\mathbf 1}_{X^{-1}(t) - T_i} - {\mathbf 1}_{X^{-1}(t) - T_{i+1}}, \ i=0,\ldots,M, \nonumber \\
X^{-1}(t) &= \dfrac{1}{q-r} \log \left[1 - (r-q)t \right].
\end{align}
Hence, when using \eqref{sub} we have
\begin{align} \label{sigmaPW_exp1}
\sigma^2(D_t, t)\Big|_{t = \Gamma_{X(t)}} &= \sum_{i=0}^M \sigma^2_i(D_t) \bar{w}_i(X(t)) = \sum_{i=0}^M \sigma^2_i(D_t) w_i(t).
\end{align}

Accordingly, if the calendar time $t$ belongs to the interval $T_0 \le t < T_1$, the infinitesimal generator $\calA$ of the semigroup $\calT^D_\nu$ is a function of $\sigma(D_t)$ (and not on $\sigma(D_\nu)$). As at $T_0 \le t < T_1$ we assume $\sigma(D) = \sigma_0(D)$, i.e., is constant in time, it doesn't depend of $\nu$. Thus, $\calA$ (which for this interval of time we will denote as $\calA_0$) is still time-homogeneous.

Similarly, one can see, that for $T_1 \le t < T_2$ the infinitesimal generator $\calA_1$ of the semigroup $\calT^D_\nu$ is also time-homogeneous and depends on $\sigma_1(D)$, etc.

\subsection{Bochner subordination}

We start with re-definition of \eqref{P0}, \eqref{P0D}. We now define the European Put value with maturity $T$ at the evaluation time $t=X(T_1)$ in the ELVG model
\begin{equation} \label{P02}
P(S_0,T_1 + T,K) = \calT^S_T [e^{-r T} P(S_0,T_1,K)].
\end{equation}
And, clearly we are interesting in the value of $T$ to be $T = T_2 - T_1$.

Similarly, we define the European Put value with maturity $\nu$ at the evaluation time $t=T_1$ in the model given by \eqref{D} as
\begin{equation} \label{P0D2}
P^D(S_0,T_1 + \nu,K) = \calT^D_\nu [e^{-r \nu} P(S_0,T_1,K)].
\end{equation}
By these definitions
\[ P(S_0,T_1 + T,K)\Big|_{T=0} = P^D(S_0,T_1 + \nu,K)\Big|_{\nu=0} = P(S_0,T_1,K)]. \]

In contrast to \eqref{Bochner2}, in case of multiple smiles at $t > T_1$ we need to change the definition of $t$ in \eqref{BI1} from $t \mapsto X(t)$ to
\begin{equation} \label{newt*}
t \mapsto X(T_1+t) - X(T_1) \equiv \Delta x(T_1,t).
\end{equation}
This definition implies two observations.

First, function $\Delta x(T_1,t)$ starts at zero at $t=0$ and is an increasing function of time. Also, in case $r=q=0$ we have $\Delta x(T_1,t) = t$. Therefore, $\Delta x(T_1,t)$ can be used as a good clock. Accordingly, similar to \eqref{gammaExp} we have
\begin{equation} \label{gammaXexp1}
\EQ[\Gamma_{\Delta x(T_1,t)}] = \Delta x(T_1,t).
\end{equation}
Second, a proof that in our model the discounted stock price is a martingale
given in Section~\ref{Process} could be repeated for times $t: \ T_1 < t \le T_2$.
When doing so, at $t > T_1$ we reset the definition of $S_t$ to
\[ S_{T_1+t} = D_{\Gamma_{\Delta x(T_1,t) }}, \quad t \ge 0. \]
Then instead of \eqref{e1} we now have
\begin{align} \label{e11}
\EQ[d S_{T_1+t}] &= \EQ[d D_{\Gamma_{\Delta x(T_1,t)}}] = \mu \EQ [D_{\Gamma_{\Delta x(T_1,t)}} d \Gamma_{\Delta x(T_1,t)}] + \EQ[\sigma(D_{\Gamma_{\Delta x(T_1,t)}}) d W_{\Gamma_{\Delta x(T_1,t) }}] \\
&= \mu \EQ [D_{\Gamma_{\Delta x(T_1,t)}}] d \Delta x(T_1,t) =
\mu \EQ [D_{\Gamma_{\Delta x(T_1,t)}}] d X(T_1+t). \nonumber
\end{align}
On the other hand,
\begin{align} \label{martCond1}
\EQ[d \left( e^{(q-r)(T_1+t)} S_{T_1+t} \right)] &=
e^{(q-r)(T_1+t)} \left\{ (q-r)  \EQ[S_{T_1+t}] d t + d \EQ[S_{T_1+t}] \right\} \\
&= e^{(q-r)t}[\mu + (q-r)S_{T_1}e^{-(r-q)T_1}] dt \nonumber
\end{align}

One can check, that with $\mu = r-q$ the RHS of \cref{martCond1} vanishes, therefore this construction can be used for option pricing.

The definition in \eqref{newt*} implies that parameter $t$ of the Gamma random clock is reset at the point $T_1$, i.e., at $0 \le t \le T_1$ it is $t \mapsto X(t) = X(t) - X(0)$, while at $T_1 < t \le T_2$ it is $t \mapsto X(T_1+t) - X(T_1)$. Using the definition of $w_i(t)$ in \eqref{sigmaPW}, this could be written as
\begin{equation} \label{t*2}
t \mapsto \sum_{i=0}^M w_i(T_i+t)[X(T_i+t) - X(T_i)]
\end{equation}
Resetting $t$ was also first proposed in \cite{CarrNadtochiy2017} but in a different form.

Then, the Bochner integral in \eqref{BI1} transforms to
\begin{equation} \label{BI3}
\calT^S_{T} P(S,T_1,K) = \int_{0}^\infty \calT^D_\nu P(S,T_1 + \nu,K) \dfrac{\nu^{m-1} e^{-\nu m /\Delta X(T_1,T)}}{(t^*)^m \Gamma(m)} d\nu.
\end{equation}
Since for a tractability reason we still want to have $m \equiv \Delta X(T_1,T)/t^* = 1$. we need to redefine $t^*$ in accordance with \eqref{t*2}. Based on that, the Bochner integral in \eqref{Bochner2} now finally reads
\begin{align} \label{Bochner22}
P(S,T_1 + T,K) &= \int_0^\infty P^D(S,T_1+\nu,K) p e^{- p \nu} d \nu,  \quad
p \equiv 1/\Delta X(T_1, T).
\end{align}

\subsection{Forward partial divided-difference equation for the second term} \label{FPDDder2}

Now we need to derive a Forward partial divided-difference equation for the second term $T_2$ similar to how this is done in Section~\ref{FPDDder}. Obviously, the Put price $P^D(S_0,T_1 + \nu,K)$ solves the same Dupire equation \eqref{Dup1}. Therefore, proceeding in the same way as in Section~\ref{FPDDder}, we apply linear differential operator $\cal L$ defined in  \eqref{Dup1} to both parts of \eqref{Bochner22}. Using time-homogeneity of $D_t$ at the interval $[T_1, T_2)$ and again the Dupire equation \eqref{Dup1}, we obtain
\begin{align} \label{b21}
-q &P(S,T_1 + T,K) - (r-q) K \calD_K P(S,T_1 + T,K) + \frac{1}{2} \sigma^2(K) \calD_K^2 P(S,T_1 + T,K) \nonumber \\
&= \int_0^\infty p e^{- p \nu} \Big[-q P^D(S,T_1 + \nu,K) - (r-q) K \calD_K P^D(S,T_1 + \nu,K)  \\
&+ \frac{1}{2} \sigma^2(K) \calD_K^2 P^D(S,T_1 + \nu,K)\Big] d \nu
= \int_0^\infty p e^{- p \nu} \calD_\nu P^D(S,T_1 + \nu,K) d \nu \nonumber \\
&= - p P^D(S,T_1,K) + p \int_0^\infty P^D(S,T_1+\nu,K) p e^{- p \nu} d\nu \nonumber \\
&= p\left[P(S,T_1 + T,K) - P^D(S,T_1,K)\right]  = p\left[P(S,T_1 + T,K) - P(S,T_1,K) \right]. \nonumber
\end{align}
Finally, taking $T = T_2  - T_1$ we obtain an ODE for the Put price $P(S,T_2,K)$.
\begin{equation} \label{finDupPut2}
\left[\frac{1}{2} \sigma^2(K) \calD_K^2 - (r-q) K \calD_K - \left(q + \dfrac{1}{
X(T_2) - X(T_1)}\right) \right] P(S,T_2,K) = - \dfrac{P(S,T_1,K)}{X(T_2) - X(T_1)}.
\end{equation}
Here the local variance function $\sigma^2(K) = \sigma^2_1(K)$ as it corresponds to the interval $(T_1,T_2]$ where the above ODE is solved.

We continue in the same way to derive an ODE for the Put price $P(S,T_i,K), \ i=1,\ldots,M$, which finally reads
\begin{equation} \label{finDupPutI}
\left[\frac{1}{2} \sigma^2(K) \calD_K^2 - (r-q) K \calD_K - \left(q + \dfrac{1}{
X(T_i) - X(T_{i-1})}\right) \right] P(S,T_i,K) = - \dfrac{P(S,T_{i-1},K)}{X(T_i) - X(T_{i-1})}.
\end{equation}
This is a recurrent equation that can be solved for all $i=1,\ldots,M$ sequentially starting with $i=1$ subject to some boundary conditions. The natural boundary conditions for the Put option price are, \cite{hull:97}
\begin{equation} \label{bc}
\begin{array}{lll}
P(S,T_i,K) = 0, & & K \to 0, \\
P(S,T_i,K) = \disD_i K - \disQ_i S \approx \disD_i K, & & K \to \infty,
\end{array}
\end{equation}
\noindent where $\disD_i = e^{-r T_i}$ is the discount factor, and $\disQ_i = e^{-q T_i}$.

A similar equation can be obtained for the Call option prices, which reads
\begin{equation} \label{finDupCallI}
\left[\frac{1}{2} \sigma^2(K) \calD_K^2 + (r-q) K \calD_K - \left(q + \dfrac{1}{
X(T_i) - X(T_{i-1})}\right) \right] C(S,T_i,K) = - \dfrac{C(S,T_{i-1},K)}{X(T_i) - X(T_{i-1})},
\end{equation}
\noindent subject to the boundary conditions
\begin{equation} \label{bcCall}
\begin{array}{lll}
C(S,T_i,K) = \disQ_i S, & & K \to 0, \\
C(S,T_i,K) = 0, & & K \to \infty.
\end{array}
\end{equation}

\section{Solution of the ODE \eqref{finDupPutI}} \label{SolutionODE}

Below we use the approach similar to \cite{ItkinLipton2017} by assuming the local variance to be a piecewise linear continuous function of strike. In contrast to \cite{ItkinLipton2017}, instead of a standard local volatility model in this paper we use the ELVG model. As the result, instead of a partial differential (Dupire) equation, we face a problem of solving the ODE in \eqref{finDupPutI}.

First, it is useful to change the dependent variable from $P(S,T_j,K)$ to
\[ V(S,T_j,K) = P(S,T_j,K) - \disD_j K, \]
\noindent which is known as a {\it covered Put}. The advantage of the covered Put is that according to \eqref{bc} its price obeys homogeneous boundary conditions.

Using this definition we now re-write \eqref{finDupPutI} in a more convenient form (while with some loose of notation)
\begin{align} \label{finODE}
& - v(x) V_{x,x}(x) + b_1 x V_x(x) + b_{0,j} V(x) = c_j, \\
b_1 &= (r - q) p_j, \quad b_{0,j} = q p_j + 1, \quad c_j = V(T_{j-1},x) + \beta x, \nonumber \\
p_j &= X(T_j) - X(T_{j-1}) > 0, \quad x = \frac{K}{S}, \quad V(x) = V(S,T_j,x), \quad v(x) = p_j\frac{\sigma^2(x)}{2 S^2}. \nonumber \\
\beta &= - S[\disD_j(1 + p_j r) - \disD_{j-1}]. \nonumber
\end{align}
In \eqref{finODE} $x$ is the inverse moneyness. In what follows we also assume that $r > q > 0$, but this assumption could be easily relaxed.

Further, suppose that for each maturity $T_j, \ j \in [1,M]$ the market quotes are provided at a set of strikes $K_i, \ i=1,\ldots,n_j$ where the strikes are assumed to be sorted in the increasing order. Then the corresponding continuous piecewise linear local variance function $v_j(x)$ on the interval $[x_{i},x_{i+1}]$ reads
\begin{equation} \label{vDef}
v_{j,i}(x) = v^0_{j,i} + v^1_{j,i} x,
\end{equation}
\noindent where we use the super-index $0$ to denote a level $v^0$, and the super-index $1$ to denote a slope $v^1$. Subindex $i=0$ in $v^0_{j,0}, v^1_{j,0}$ corresponds to the interval $(0, x_1]$. Since $v_j(x)$ is continuous, we have
\begin{equation} \label{cont}
v^0_{j,i} + v^1_{j,i} x_{i+1} = v^0_{j,i+1} + v^1_{j,i+1} x_{i+1}, \quad i=0,\ldots,n_j-1.
\end{equation}
The first derivative of $v_j(x)$ experiences a jump at points $x_{i}, \irg{i}{1}{n_j}$. We also assume that $v(x,T)$ is a piecewise constant function of time, i.e., $v^0_{j,i}, v^1_{j,i}$ don't depend on $T$ on the intervals $[T_j, T_{j+1}), \ j \in [0,M-1]$, and jump to new values at the points $T_j, \irg{j}{1}{M}$.

With the above assumptions in mind, \eqref{finODE} can be solved by induction. One starts with $T_0 = 0$, and on each time interval $[T_{j-1},T_j], \ \irg{j}{1}{M}$ solves the problem \eqref{finODE} for $V(x) \mapsto P(S,T_j,x) - d_j S x$.

Since $v(x)$ is a piecewise linear function, the solution of \eqref{finODE} can also be constructed separately for each interval $[x_{i-1},x_i]$. By taking into account the explicit representation of $v(x)$ in \eqref{vDef}, from \eqref{finODE} for the $i$-th spatial interval we obtain
\begin{align} \label{Laplace2}
-(b_2 + a_2 x) V_{xx}(x) &+ b_1 x V_x(x) + b_0 V(x) = c \\
b_2 &= v^0_{j,i}, \ a_2 = v^1_{j,i}. \nonumber
\end{align}
We proceed by introducing a new independent variable $z = (b_2 + a_2 x)b_1/a_2^2, \ z \in \mathbb{R}^+$, so that \eqref{Laplace2} transforms to
\begin{align} \label{Laplace3}
-z V_{zz}(z) &+ (z-q_2)V_z(z) + q_1 V(z) = \chi \\
q_1 &= b_0/b_1, \ q_2 = b_2 b_1/a_2^2, \ \chi = c/b_1. \nonumber
\end{align}

The \eqref{Laplace3} is an {\it inhomogeneous} Laplace equation, \cite{PolyaninSaitsevODE2003}, page 155. It is well known that if $y_1=y_1(z)$, $y_2=y_2(z)$ are two fundamental solutions of the corresponding {\it homogeneous} equation, then the general solution of \eqref{Laplace3} can be represented as
\begin{align} \label{solInhom}
V(z) &= C_1 y_1(z) + C_2 y_2(z) + \frac{1}{b_1} I_{12}(z) \\
I_{12}(z) &= -y_2(z) \int \dfrac{ y_1(z) f(z)} {W z}d z + y_1(z) \int  \dfrac{ y_2(z) f(z)}{W z}d z \equiv I_1 + I_2, \nonumber \\
f(z) &= V(T_{j-1},z) - k_1 - k_2 z, \quad
k_1 = \beta \frac{b_2}{a_2}, \quad k_2 = - \beta \frac{a_2}{b_1}, \nonumber
\end{align}
\noindent where $W = y_1 (y_2)_z - y_2 (y_1)_z$ is the so-called Wronskian, and $\beta$ is defined in \eqref{finODE}. Then the problem is to determine suitable fundamental solutions of the homogeneous Laplace equations. Based on \cite{PolyaninSaitsevODE2003}, if $a_2 \ne 0$, they read
\begin{equation} \label{homog}
y_i(z) = {\mathcal V}_i (q_1, q_2, z), \quad i=1,2
\end{equation}
Here ${\mathcal V}_i(a,b,z)$ is an arbitrary solution of the degenerate hypergeometric equation, i.e., Kummer's function, \cite{as64}. Two types of Kummer's functions are known, namely $M(a,b,z)$ and $U(a,b,z)$, which are Kummer's functions of the first and second kind.

It is known, that there exist several pairs of such independent solutions. Therefore, for every spatial interval in $z$ among all possible fundamental pairs we have to determine just one which is numerically satisfactory at this interval (see \cite{Olver1997} for the detailed definition of satisfactory solutions and the corresponding discussion). Since our boundary conditions are set at zero and positive infinity, we need a numerically satisfactory solution for the positive half of the real line.

Similar to \cite{ItkinLipton2017}, in the vicinity of the origin we choose the numerically satisfactory pair as, \cite{Olver1997}
\begin{align} \label{Kummer0}
y_1(\chi) &= M\left(q_1, q_2, z\right) = e^z M\left(q_2-q_1, q_2, -z\right), \\
y_2(\chi) &= z^{1-q_2} M\left(q_1 - q_2 + 1, 2-q_2, z\right) =
z^{1-q_2} e^z M\left(1-q_1, 2-q_2, -z\right), \nonumber \\
W &= \sin(\pi q_2) z^{-q_2} e^{z}/\pi. \nonumber
\end{align}
However, in the vicinity of infinity the numerically satisfactory pair is, \cite{Olver1997}
\begin{align} \label{Kummer1}
y_1(\chi) &=   U\left(q_1, q_2, z\right) = z^{1-q_2} U\left(q_1-q_2+1, 2-q_2, z\right), \\
y_2(\chi) &= e^z U\left(q_2 - q_1, q_2, -z\right) = e^z z^{1-q_2} U\left(1-q_1, 2-q_2, -z\right), \nonumber \\
W &= (-1)^{q_1 - q_2} e^{z} z^{-q_2}. \nonumber
\end{align}

As two solutions $J_1(q_1,q_2, z), J_2(q_1,q_2, z)$ are independent, \eqref{solInhom} is a general solution of \eqref{Laplace3}. Two constants $C_1,C_2$ should be determined based on the boundary conditions for the function $V(z)$.

The boundary conditions for the ODE \eqref{Laplace2} in a strike $K$ space (or in $x$ space) should be set at zero and infinity. Based on the usual shape of the local variance curve and its positivity, for $x \to 0$, we expect that $v^1_{j,i} < 0$.  Similarly, for $x \to \infty$ we expect that $v^1_{j,i} > 0$. In between these two limits the local variance curve for a given maturity $T_j$ is assumed to be continuous, but the slope of the curve could be both positive and negative, see, e.g., \cite{ItkinSigmoid2015} and references therein. Also, by definition $z = v_{j,i}$, and $\mathrm{Dom}(z) = \mathbb{R}^+$. Thus, at high strikes $a_2 = v^1_{j,i} > 0$. Therefore, the boundary conditions for \eqref{Laplace3} should be set at $z = b_2$ (which corresponds to the boundary $K=0$) and at $z \to \infty$. These are the boundary conditions given in \eqref{bc}.

\section{Computation of the source term} \label{SolutionInt}

Computation of the source term $p I_{12}$ in \eqref{solInhom} could be achieved in several ways. The most straightforward one is to use numerical integration as the Put price $P(x,T_{i-1})$ as a function of $x$ is already known when we solve \eqref{finODE} for $T=T_i$. However, as this is discussed in detail in \cite{ItkinLipton2017}, function $P(x,T_{i-1})$ is known only for a discrete set of points in $x$. Therefore, some kind of interpolation is necessary to find its values at the other points.

\subsection{No-arbitrage interpolation}

As shown in \cite{ItkinLipton2017}, this interpolation must preserve no-arbitrage. So, for instance, a standard linear interpolation is not a good candidate, since its violates no-arbitrage conditions. Indeed, given three Put option prices $P(K_1), P(K_2), P(K_3)$ for three strikes $K_1 < K_2 < K_3$, the necessary and sufficient conditions for an arbitrage-free system read, \cite{CoxRubinstein1985}
\begin{align} \label{noarb}
P(K_3) &> 0, \qquad P(K_2) < P(K_3), \\
Bs = (K_3 - K_2)P(K_1) &- (K_3 - K_1)P(K_2) + (K_2 - K_1)P(K_3) > 0. \nonumber
\end{align}
Suppose that we want to use linear interpolation in the strike space on the interval $[K_1,K_3]$ to find the unknown Put option price $P(K_2)$ given the values of $P(K_1), P(K_3)$,
\begin{equation*}
P(K_2) \equiv P_l(K_2) = \dfrac{P(K_1) K_3 - P(K_3) K_1}{K_3 - K_1} + \dfrac{P(K_3) - P(K_1)}{K_3 - K_1} K_2.
\end{equation*}
When plugging this expression into the second line of \eqref{noarb}, the left hands side of the latter vanishes, so the third no-arbitrage condition is violated.

In \cite{ItkinLipton2017} it is shown that this problem could be resolved if we use linear interpolation with a modified independent variable (further on we denote it as $P_F(K)$),
\begin{align} \label{lin2}
P(K_2) &\equiv P_F(K_2) \\
&= \dfrac{P(K_1) f(K_3) - P(K_3) f(K_1)}{f(K_3) - f(K_1)} + \dfrac{P(K_3) - P(K_1)}{f(K_3) - f(K_1)} f(K_2), \nonumber
\end{align}
\noindent where $f(K)$ is a convex and increasing function in $[K_1,K_3]$. Indeed, if $f(K)$ is convex, then $P(K_2) = P_F(K_2) = P_l(K_2) - \varepsilon, \ \varepsilon > 0$ (see Fig.~2 in \cite{ItkinLipton2017}). Substitution of this expression into the second line of \eqref{noarb} gives $(K_3 - K_1) \varepsilon > 0$, which is true. The second condition in \eqref{noarb} now reads
\[ (P(K_1) - P(K_3))(f(K_3) - f(K_2))(f(K_1) - f(K_3)) > 0, \]
\noindent which is also true since $f(K)$ is an increasing function of $K$.

Alternatively, one can use non-linear interpolation. In \cite{ItkinLipton2017}) both approaches were combined, and it was proved that the new interpolation scheme preserves no-arbitrage. Moreover, the final representation of the modified Put price (which is a dependent variable in their approach) obtains a nice tractable representation, so the integral $I_{12}$ can be computed in closed form. Here we want to exploit the same idea, thus significantly improving performance of our model as compared with the numerical integration.

Therefore, here we propose the following interpolation scheme
\begin{align} \label{linNew}
P(x) &\equiv P_F(x) = \gamma_1 + \gamma_2 x^2, \quad x_1 \le x \le x_3, \\
\gamma_1 &= \dfrac{P(x_3) x_1^2   - P(x_1) x^2_3}{x_1^2 - x_3^2}, \qquad
\gamma_2 = \dfrac{P(x_1) - P(x_3)}{x_1^2 - x_3^2}. \nonumber
\end{align}
Then Proposition similar to that in \cite{ItkinLipton2017} can be proved.
\begin{proposition} \label{prop1}
The interpolation scheme in \eqref{linNew} is arbitrage free at the interval $[K_1,K_3]$.
\end{proposition}
\begin{proof}
Observe, that the no-arbitrage conditions in \eqref{noarb} are discrete versions of the conditions
\[ P > 0, \quad P_K > 0, \quad P_{K,K} > 0, \]
They, in turn, correspond to the conditions
\[ P > 0, \quad P_x > 0, \quad P_{x,x} > 0, \]
\noindent as $x'(K) = 1/S > 0$. By differentiating the first line of \eqref{linNew} one can check that the proposed interpolation obeys these conditions provided that $P$ is an increasing function of $K$ (or $x$) given the values of all other parameters to be constant. For instance, this is the case for the Black-Scholes Puts.
\bs
\end{proof}

As by definition $z$ is a linear function of $x$, a similar interpolation scheme can be used in the $z$ space, with a similar proof of no-arbitrage.

\subsection{No-arbitrage at consecutive intervals}

Proposition~\ref{prop1} guarantees that the proposed interpolation doesn't introduce an arbitrage into the solution if any three strikes belong to the same interval $[K_1,K_3]$. However, what if we consider strikes $K_2, K_3, K_4$ as this is schematically  depicted in Fig.~\ref{str2interv}.
\begin{figure}[H]
\begin{center}
\begin{tikzpicture}[line/.style={<->},thick, framed, scale=1.4,
redline/.style={shape=rectangle, draw=red, line width=2},
blueline/.style={shape=rectangle, draw=blue, line width=2}
]

\draw[->] (-3.0,0) -- (5,0) node[right] {$K$};
\draw[->] (-3.0,-0.2) -- (-3.0,4.5) node[right] {$P(K)$};
\draw[red,ultra thick] (-2,1) parabola (4.5,3.5);
\node at (-3.0,-0.3) {$0$};

\node at (-2,-0.3) {$K_1$};
\node at (-2,1.3) {$P_1$};
\node at (-2,1.) {$\bullet$};
\draw[red, dashed] (-2,0) -- (-2,1.);

\node at (-0,-0.3) {$K_2$};
\node at (0.2,0.6) {$P_2$};
\node at (-0,0.8) {$\bullet$};
\draw[red, dashed] (-0,0) -- (-0,0.8);

\node at (1.5,-0.3) {$K_3$};
\node at (1.5,2.) {$P_3$};
\node at (1.5,1.7) {$\bullet$};
\draw[red, dashed] (1.5,0) -- (1.5,1.7);

\node at (2.5,-0.3) {$K_4$};
\node at (2.7,1.5) {$P_4$};
\node at (2.5,1.65) {$\bullet$};
\draw[red, dashed] (2.5,0) -- (2.5,1.65);

\node at (4.5,-0.3) {$K_5$};
\node at (4.5,3.8) {$P_5$};
\node at (4.5,3.5) {$\bullet$};
\draw[red, dashed] (4.5,0) -- (4.5,3.5);

\draw[blue,ultra thick] (-2,1.) parabola bend (-0.2, 0.8) (1.5,1.7);
\draw[blue,ultra thick] (1.5,1.7) parabola bend (2., 1.6) (4.5,3.5);

\matrix [draw,below left] at (1.5, 4.5) {
  \node [redline,label=right:Exact solution] {}; \\
  \node [blueline,label=right:Interpolation] {}; \\
};

\end{tikzpicture}
\end{center}
\caption{Three strikes $K_2, K_3, K_4$ which belong to the consecutive intervals.}
\label{str2interv}
\end{figure}
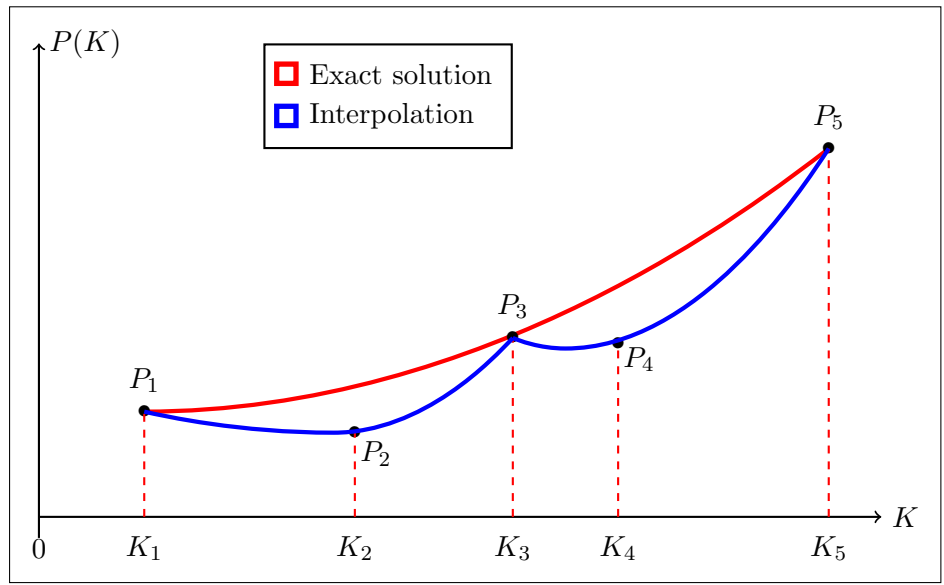

Here at the interval $[K_1,K_3]$ the exact solution is depicted by the red line, while our quadratic interpolation is in blue. Accordingly, the Put prices $P_1,P_3$ are the market quotes, so they assumed to be the exact prices with no market arbitrage. By our construction, these prices also don't have a model arbitrage. At the consecutive interval $[K_3,K_5]$ a similar construction applies.

We have to emphasize that this graph is pure illustrative, and no-arbitrage interpolation guarantees  that $P'(K) > 0$, while the blue line in Fig.~\ref{str2interv} doesn't support this. However, if we draw an accurate picture by using the above formulae, it would be almost impossible to distinguish the red and blue lines. Therefore, we changed convexity and skew of the blue line to make the difference visible.

By Proposition~\ref{prop1} given a set of strikes $K_1, K_2, K_3$ the price $P_2$ obtained by interpolation preserves no-arbitrage. The same is true for $P_5$ given the Put prices $P_3, P_5$ at strikes $K_3, K_5$.
Now assume that given $K_1, K_3, K_5$ and $P_1, P_3, P_5$ we want to check the no-arbitrage conditions for the set of strikes $K_2, K_3, K_4$. The Proposition~\ref{prop1} doesn't help in this situation, so we need a special consideration of this case.

Obviously, the first and second conditions in \eqref{noarb} are still satisfied in this case, so we need to check that the butterfly spread is positive. Unfortunately, at the moment we don't have a general analytical solution of this problem, while some particular cases can be addressed. Thus this remains an open question. However, we checked this condition numerically. In doing so we used the Black-Scholes Put prices $P_1, P_3, P_5$\footnote{This is done to preserve upper bounds on the Put price that $P(S, K,T,r) \le K e^{-rT}$, \cite{Levy1985}.} and built a 2D plot of $Bs$ which is the left-hands side of the third line in \eqref{noarb}. The results for two cases presented in Table~\ref{tabCases} are presented in Fig.~\ref{intCase1}, \ref{intCase2}.

\begin{table}[H]
\begin{center}
\begin{tabular}{|r|r|r|r|r|r|r|r|}
\hline
Test & $S$ & $r$ & $\sigma_{BS}$ & $T$ & $K_1$ & $K_3$ & $K_5$ \\
\hline
1 & 100 & 0.01 & 0.5 & 2 & 80 & 100 & 130 \\
\hline
2 & 100 & 0.1 & 0.1 & 0.1 & 90 & 100 & 105 \\
\hline
\end{tabular}
\caption{Parameters of the test for non-negativity of the Butterfly spread. $\sigma_{BS}$ is the Black-Scholes implied volatility.}
\label{tabCases}
\end{center}
\end{table}

Overall, we ran a lot of tests and didn't find any case where the butterfly spread would become negative. This partly supports our no-arbitrage interpolation. More sophisticated cases where, e.g., instead of strike $K_3$ in the butterfly spread at strikes $K_2, K_3, K_4$ we use another strike $K_6$ such that $K_1 < K_2 < K_3 < K_4 < K_6 < K_5$, could be treated in a similar way. Again, our numerical tests didn't reveal any case where a butterfly spread would become negative.

\begin{figure}[h!]
\begin{minipage}{0.46\linewidth}
\begin{center}
\fbox{\includegraphics[width=\linewidth]{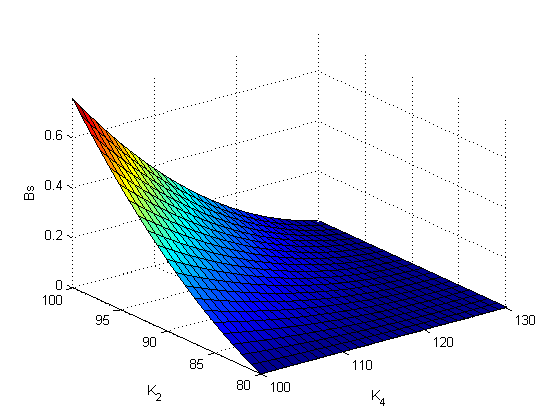}}
\caption{Butterfly spread $Bs$ for a set of strikes $K_2, K_3, K_4$ computed in Test 1 in Table~\ref{tabCases}.}
\label{intCase1}
\end{center}
\end{minipage}
\hspace{0.04\linewidth}
\begin{minipage}{0.46\linewidth}
\begin{center}
\fbox{\includegraphics[width=\linewidth]{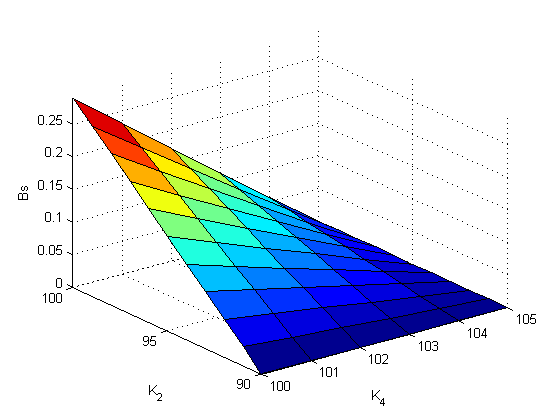}}
\caption{Butterfly spread $Bs$ for a set of strikes $K_2, K_3, K_4$ computed in Test 2 in Table~\ref{tabCases}.}
\label{intCase2}
\end{center}
\end{minipage}
\end{figure}

\subsection{Computing the integrals in \eqref{solInhom} far from $z = 0$}

Using the interpolation scheme proposed in above, consider the first integral in \eqref{solInhom}.  To remind, we compute it at some interval $z \in [z_i,z_{i+1}], \  \irg{i}{1}{n_j}$. Picking together the solutions in \eqref{Kummer0} with the interpolation scheme for $P(z,T_{j-1})$ and Wronskians in \eqref{Kummer0}, and substituting them into the first integral in \eqref{solInhom} we obtain
\begin{align} \label{Int1}
\int & \dfrac{ y_2(z) f(z,T_{j-1})}{W z}d z = A \Big[-B_0
+ B_1 M(-2-q_1, -1-q_2, -z)
+ B_2 M(-1-q_1, -q_2, -z) \\
&+ B_3 M(-q_1, 1-q_2, -z)\Big], \nonumber \\
A &= \dfrac{1}{b_1^2 q_1} \pi(1-q_2) \csc(\pi q_2), \nonumber \\
B_0 &= \frac{1}{a_2^2 (q_1+1) (q_1+2)} \Big[
a_2 b_1 (q_1+2) \left(a_2^2 \beta q_2 - b_1 (q_1+1) (\beta  b_2-a_2 \gamma_1)
\right) \nonumber \\
&+ \gamma_2 \left(2 a_2^4 q_2 (q_2+1)-2 a_2^2 b_1 b_2 (q_1+2)q_2 + b_1^2 b_2^2 (q_1+1) (q_1+2) \right) \Big],
\nonumber \\
B_1 &= 2 a_2^2 \gamma_2 \frac{q_2 (q_2+1)}{(1+q_1)(2+q_1)}, \nonumber \\
B_2 &= (a_2 b_1 \beta - 2 b_1 b_2 \gamma_2 + 2 a_2^2 \gamma_2 z)  \frac{q_2}{1+q_1}, \nonumber \\
B_3 &= \frac{1}{a_2^2}
\left[a_2 b_1\left(a_2^2 \beta  z + a_2 b_1 \gamma_1 - \beta b_1 b_2 \right)
 + \gamma_2 \left(b_1 b_2 - a_2^2 z\right)^2\right]. \nonumber
\end{align}
Similarly
\begin{align} \label{Int2}
\int & \dfrac{ y_1(z) f(z,T_{j-1})}{W z}d z = \bar{A} \Big[
\bar{B}_1 M(q_2-q_1, 1+q_2, -z)
+ \bar{B}_2 \ \pFq{2}{2}{q_2-q_1, 1+q_2; q_2,2+q_2; -z} \nonumber \\
&+ \bar{B}_3 \ \pFq{2}{2}{q_2-q_1, 2+q_2; q_2,3+q_2; -z} \Big], \\
A &= \pi z^{q_2} \csc(\pi q_2) \Gamma(q_2), \nonumber \\
B_1 &=
\frac{a_2^2 \gamma_1 - a_2 \beta  b_2 + b_2^2 \gamma_2}{a_2^2 \Gamma(1+q_2)}, \nonumber \\
B_2 &=\frac{\Gamma(q_2+1) }{\Gamma(q_2) \Gamma(2+q_2) b_1} (a_2 \beta -2 b_2 \gamma_2) z, \nonumber \\
B_3 &= \frac{\Gamma(q_2+1) }{\Gamma(q_2) \Gamma(3+q_2) b_1^2}
a_2^2(1+q_2)\gamma_2 z^2, \nonumber
\end{align}
\noindent where $\pFq{p}{q}{{a_1,...,a_p}; {b_1,...,b_q}; z}$ is the generalized hypergeometric function, \cite{Olver1997}.

\subsection{Computing the integrals in \eqref{solInhom} far from $z = \pm \infty$}

Here we proceed in the same way as in the previous section. Again, we pick together the solutions in \eqref{Kummer1} with the interpolation scheme for $P(z,T_{j-1})$ and Wronskians in \eqref{Kummer1}, and substitute them into the first integral in \eqref{solInhom} we obtain
\begin{align} \label{Int11}
\int & \dfrac{ y_2(z) f(z,T_{j-1})}{W z}d z = (-1)^{q_2 - q_1} [C_0 J_0 + C_1 J_1 + C_2 J_2], \\
J_i &= \int z^i U(1-q_1, 2-q_2, -z) dz, \nonumber \\
C_0 &= \frac{b_2^2 \gamma_2}{a_2^2} - \frac{\beta b_2}{a_2} + \gamma_1, \quad
C_1 = \frac{a_2 \beta - 2 b_2 \gamma_2}{b_1}, \quad
C_2 = \frac{a_2^2 \gamma_2}{b_1^2}. \nonumber
\end{align}
It is known, \cite{as64}, that
\begin{equation*}
J_0 = - \frac{1}{q_1}U(-q_1, 1-q_2, -z).
\end{equation*}
Then, $J_1, J_2$ can be found using integration by parts to yield
\begin{align*}
J_1 &= z J_0 + \frac{1}{q_1(1+q_1)}U(-1-q_1, -q_2, -z), \\
J_2 &= z J_1 + \frac{1}{q_1(2+3 q_1 + q_1^2)}U(-2-q_1, -1-q_2, -z)
 - z\frac{1}{q_1(1 + q_1)}U(-1-q_1,-q_2,-z).
\end{align*}
Similarly
\begin{align} \label{Int21}
\int & \dfrac{ y_1(z) f(z,T_{j-1})}{W z}d z = (-1)^{q_2 - q_1} [C_0
{\mathcal J}_0 + C_1 {\mathcal J}_1 + C_2 {\mathcal J}_2], \\
{\cal J}_i &= \int z^i e^{-z} U(1+q_1-q_2, 2-q_2, z) dz. \nonumber
\end{align}
The integrals ${\mathcal J}_i $ have been considered in \cite{ItkinLipton2017} using the approach of \cite{kummerInt1970}. Borrowing from there the result
\begin{equation*}
{\mathcal J}_0 = \int e^{-z} U(1+q_1-q_2,2-q_2,z) d z = -e^{-z} U(q_1-q_2,1-q_2,z),
\end{equation*}
\noindent and using integration by parts, we obtain
\begin{align*}
{\mathcal J}_1 &= z {\mathcal J}_0 + e^{-z}U(q_1-q_2-1, -1-q_2, z), \\
{\mathcal J}_2 &= z {\mathcal J}_1 - \int  {\mathcal J}_1 dz = (z-1) {\mathcal J}_1 - \int e^{-z}U(q_1-q_2-1, -1-q_2, z) dz \\
&= (z-1) {\mathcal J}_1 + e^{-z}U(q_1-q_2-2, -2-q_2, z).
\end{align*}

\subsection{Some additional notes}

Based on the no-arbitrage interpolation and some analytics proposed in this Section, we managed to find the solution \eqref{solInhom} of the forward equation  \eqref{finODE} in closed form. This solution by construction is arbitrage free at any interval where the local variance function defined in \eqref{vDef} is linear. In other words we proved, that if we consider, say 3 strikes $0 < K_1 < K_2 < K_3 < \infty$ such that, e.g., $x_1  = K_1/S \in [x_i, x_{i+1}]$,  $x_2  = K_2/S \in [x_i, x_{i+1}]$,  $x_3  = K_3/S \in [x_{i}, x_{i+1}]$, then the solution at these 3 points obeys no-arbitrage conditions.

\section{Calibration of smile for a given term $T_i$} \label{calib}

Calibration problem for the local volatility model can be formulated as follows: given market quotes of Call and/or Put options corresponding to a set of $N$ strikes $\{K\}:= K_j, \ j \in [1,N]$ and same maturity $T_i$, find the local variance function $\sigma(K)$ such that these quotes solve equations in \eqref{finDupPutI}, \eqref{finDupCallI}.

As mentioned in \cite{ItkinLipton2017}, there are two main approaches to solving this problem. The first approach attempts to construct a continuous implied volatility (IV) surface matching the market quotes by using either some parametric or non-parametric regression, and then generates the corresponding LV surface via the well-known relationship between the local and implied variances also known as the Dupire formula, see, e.g., \cite{ItkinSigmoid2015} and references therein. To be practically useful, this construction should guarantee no arbitrage for all strikes and maturities, which is a serious challenge for any model based on interpolation. If the no-arbitrage condition is satisfied, then the LV surface can be calculated using the Dupire formula. The second approach relies on the direct solution of the corresponding forward equation (which is the Dupire equation in the Black-Scholes world, or \eqref{finDupCallI}, \eqref{finDupPutI} in our model) using either analytical or numerical methods. The advantage of this approach is that it guarantees no-arbitrage. However, the problem of the direct solution can be ill-posed, \cite{Coleman2001}, and is rather computationally intensive.

In this Section we show that the second approach could be significantly simplified when using the ELVG model, so calibration of the smile could be done very fast and accurate.

Further, for the sake of certainty, suppose that all known market quotes are Puts, despite this can be easily relaxed. Also, suppose that the shape of a local variance is given by some function $\sigma_j(K) = f_j(K, p_1,\ldots,p_L)$, where $p_1,\ldots,p_L$ is a set of the model parameters to be determined. For instance, in \cite{LiptonSepp2011, CarrNadtochiy2017} the local variance is assumed to be a piecewise constant function of strike, while in \cite{ItkinLipton2017} this is a piecewise linear function of strike.

In this paper we also assume the local variance to be a piecewise linear function of strike. Moreover, for our model we obtained a closed form representation of the Put option prices via parameters of the model given in Sections~\ref{SolutionODE},\ref{SolutionInt}. Therefore, calibration of the model to the given set of smiles could be provided as follows. First, using the above-mentioned closed form solution for a fixed interval in $x$ where parameters of the model are constant, we construct the combined solution valid for all $x \in \mathbb{R}^+$. At the second step, the parameters of the local variance function $v^0_{j,i}, v^1_{j,i}$ can be found together with the integration constants $C_1, C_2$ in \eqref{solInhom} by solving a system of non-linear algebraic equations.

\subsection{The combined solution in $x \in \mathbb{R}^+$ \label{wholeSol}}

Suppose that the Put prices for $T=T_j$ are known for $n_j$ ordered strikes. The location of these strikes on the $x$ line is schematically depicted in Fig.~\ref{Fig3}.
\begin{figure}[H]
\begin{center}
\begin{tikzpicture}[line/.style={<->},thick, framed, scale=1.4]

\draw[->] (-3.0,0) -- (5,0) node[right] {$x$};
\draw[->] (-3.0,-0.2) -- (-3.0,3.5) node[right] {$v(x)$};
\draw[red,ultra thick] (0.5,0.5) parabola (4.5,1.5);
\draw[red,ultra thick] (0.5,0.5) parabola (-2,3);
\node at (-3.0,-0.3) {$0$};
\node at (-2,-0.3) {$x_1$};
\node at (-2,3.) {$\bullet$};
\draw[red, dashed] (-2,0) -- (-2,3.);
\node at (-2.5,0.5) {$B_1$};
\draw[blue, dashed] (-2,3.) -- (-2.8,3.2);
\node at (-1,-0.3) {$x_2$};
\node at (-1,1.4) {$\bullet$};
\draw[red, dashed] (-1,0) -- (-1,1.4);
\node at (-1.5,0.5) {$B_{12}$};
\draw[blue, dashed] (-2,3.) -- (-1,1.4);
\node at (1.5,-0.3) {$x_3$};
\node at (1.5,0.55) {$\bullet$};
\draw[red, dashed] (1.5,0) -- (1.5,0.55);
\node at (-0.5,0.5) {$B_{23}$};
\draw[blue, dashed] (-1,1.4) -- (1.5,0.55);
\node at (2.5,-0.3) {$\ldots$};
\node at (3.5,-0.3) {$x_{n_j}$};
\node at (3.5,1.05) {$\bullet$};
\draw[red, dashed] (3.5,0) -- (3.5,1.05);
\draw[blue, dashed] (1.5,0.55) -- (3.5,1.05);
\draw[blue, dashed] (3.5,1.05) -- (4.5,2.);
\node at (4,0.5) {$B_{n_j}$};
\node at (0.5,1.3) {$2$};
\node at (0.7,0.3) {$1$};

\end{tikzpicture}
\end{center}
\caption{Construction of the combined solution in $x \in \mathbb{R}^+$: 1 (red solid line - the real (unknown) local variance curve, 2 (dashed blue line) - a piecewise linear solution.}
\label{Fig3}
\end{figure}
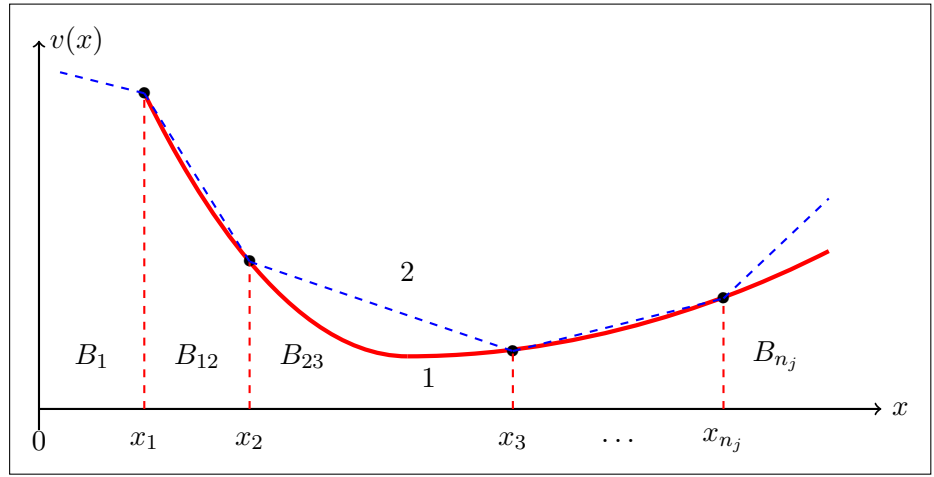

Recall, that the Put prices are given by \eqref{solInhom}, which in a more convenient form at the interval $x_{i-1} \le x \le x_i$ and at $T=T_j$ can be represented as
\begin{align} \label{combPut}
P_i(x) &= \CO{i}{j} {\cal J}_1(q_1, q_2, z)
+ \CT{i}{j} {\cal J}_2(q_1, q_2, z) + \frac{1}{b_1}I_{12}(z) + \disD_j K, \\
z &\equiv (b_2 + a_2 x)b_1/a_2^2 = (v_{j,i}^0 + v_{j,i}^1 x)b_1/a_2^2. \nonumber
\end{align}
Here, for consistency we change notation of two integration constants which belong to the $i$-th interval in $x$ and $j$-th maturity to $\CO{i}{j},\CT{i}{j}$.

For the open interval $B_1$ in Fig.~\ref{Fig3}, since function $K_\nu(z)$ diverges when  $z \to 0$, we have to put $\CO{1}{j} = 0$ as the boundary condition\footnote{
Actually, since $x \to 0$ implies $z = v \to b_2$, so $b_2$ should be non-negative, $b_2 \ge 0$. Therefore, the only case when $z \to 0$ at $x \to 0$ is when $b_2 = 0$.}. Therefore, \eqref{combPut} contains just one yet unknown constant $\CT{1}{j}$. For the closed intervals $x \in [x_{i-1}, x_{i}], \ i \in [2,n_j]$ the solutions in \eqref{combPut} have two yet unknown constants $\CO{i}{j}, \CT{i}{j}$, since $x$ is finite on the corresponding intervals, and both solutions $y_1(x), y_2(x)$ are well-behaved. Finally, for the interval $x \in [x_{n_j}, \infty)$, according to the boundary conditions in \eqref{bc} we must set $C^{(2)}_{2,n_j+1} = 0$.

Rigorously speaking, we also have to show that in the limits $x \to 0$ and $x \to \infty$ the source term $I_{12}(z)$ in \eqref{solInhom} also vanishes. This could be done similar to Proposition~2 in \cite{ItkinLipton2017}.

Thus, we have $2 n_j$ unknown constants to be determined. Since the local volatility function $v_i$ is continuous at the points $x_i, \ i=1,\ldots,n_j$, so should be the Put options prices $P(x,T_j)$. Therefore, we require that at the points $x_i, \ i = 1,...,n_j$ the solution for Puts and its first derivative in $x$ should be a continuous function of $x$. Thus, if the local variance function is known, the above constants solve a system of $2 n_j$ algebraic equations. This system has a block-diagonal structure where each block is a 2x2 matrix. Therefore, it can be easily solved with the linear complexity $O(n_j)$.

When computing the first derivatives, we take into account that, \cite{as64}
\begin{align} \label{der}
\fp{M(a,b,z)}{z} &= \dfrac{a}{b}M(a+1,b+1,z), \quad
\fp{U(a,b,z)}{z} = -a U(a+1,b+1,z), \\
\partial_z I_{12}(z) &= \left[\frac{y'_{1}(z)}{y_{1}(z)} I_1 + \frac{y'_{2}(z)}{y_{2}(z)} I_2 \right] a_2. \nonumber
\end{align}
Therefore, computing the derivatives of the solution doesn't cause any new technical problem.

\subsection{Additional equations for calibration}

As we have already mentioned above, the standard way of doing calibration of the local volatility model would be that described, e.g., in \cite{ItkinLipton2017}. Namely, given the maturity $T_j$ and some initial guess of the local variance parameters $v_{j,i}^0, v_{j,i}^1, \ \forall i \in [1,n_j]$, the following steps represented in Panel~\ref{Algo} have to be achieved, e.g., in the standard least-square method,
\begin{center}
\begin{algorithm}[H]
 \KwIn{Strikes $z_i, \ i \in [1,n_j]$, Put prices  $V^{market}_i, \ i \in [1,n_j]$}
 \KwOut{$v_{j,i}^0, v_{j,i}^1, \ \forall i \in [1,n_j]$}
 \emph{{\bf Initialization}: The initial guess of $v_{j,i}^0, v_{j,i}^1, \ \forall i \in [1,n_j]$, the tolerance $\epsilon$} \;
\While{1}{
  1. Solve the system for $\CO{i}{j}, \CT{i}{j}$ \;
  2. Compute Put option prices $V(x)$\;
  3. Compute the total error $\Delta = \sum_{i=1}^{n_j} [V(x_i) - V^{market}(x_i)]^2$\;
    \eIf{ $\Delta > \epsilon$}{
        New guess for $v_{j,i}^0, v_{j,i}^1, \ \forall i \in [1,n_j]$\;
    }{
        break\;
    }
 }
\caption{Calibration of the local volatility model using a least-square method.}
\label{Algo}
\end{algorithm}
\end{center}
\vspace{0.2in}
Here $V^{market}(z_i)$ are the market Put quotes at the given strikes and maturity. Obviously, when the number of calibration parameters (strikes) is high, this algorithm is slow even if the closed form solution is known and can be used at Step~2. Things become even worse when a numerical solution at Step~2 has to be used if the closed form solution is not available.

However, in our case this tedious algorithm can be fully eliminated. Indeed, at every point $i$ in strike space, $i \in [1,n_j]$ we have four unknown variables $v_{j,i}^0, v_{j,i}^1, \CO{i}{j}, \CT{i}{j}$. We also have four equations which contain these variables, namely
\begin{align} \label{finSystem}
P_i(x)|_{x=x_i} &= P_{i+1}(x)|_{x = x_i}, \\
P_i(x)|_{x=x_i} &= P_{market}(x_i), \nonumber \\
\fp{P_{i+1}(x)}{x}\Big|_{x = x_i} &=  \fp{P_{i}(x)}{x}\Big|_{x = x_i}, \nonumber \\
v^0_{j,i} + v^1_{j,i} x_{i} &= v^0_{j,i+1} + v^1_{j,i+1} x_{i}, \quad i=1,\ldots,n_j. \nonumber
\end{align}
Also, based on \eqref{cont}, the last line in \eqref{finSystem} could be re-written as a recurrent expression
\begin{equation} \label{recurr}
v^0_{j,i} = v^0_{j,n_j} + \sum_{k=i+1}^{n_j} x_{k}(v^1_{j,k} - v^1_{j,k-1}) , \quad i=0,\ldots,n_j-1.
\end{equation}

The \eqref{finSystem} is a system of $4 n_j$ nonlinear equations with respect to $4 (n_j+1)$ variables $v_{j,i}^0$, $v_{j,i}^1$, $\CO{i}{j}$, $\CT{i}{j}$. We remind that according to the boundary conditions $\CO{1}{j} = \CT{n_j}{j} = 0$. Therefore, we need two additional conditions to provide a unique solution. For instance, often traders have an intuition about the asymptotic behavior of the volatility surface at infinity, which, according to our construction, is determined by $v^1_{j,n_j}$ and $v^1_{j,0}$.

Overall, solving the nonlinear system of equations \eqref{finSystem} provides the final solution of our problem. This can be done by using standard methods, and, thus, no any optimization procedure is necessary. However, a good initial guess still would be helpful for a better (and faster) convergence.

\subsection{Smart initial guess}

The initial guess of the solution of \eqref{combPut} can be constructed, for instance, as follows. We take advantage of the fact that according to \eqref{finODE} the local variance function $v(x)$ could be explicitly expressed as
\begin{equation} \label{localVarappr}
v(x) = \dfrac{b_1 x V_x(x) + b_0 V(x) - c}{V_{x,x}(x)}.
\end{equation}

Given maturity $T_j$ and approximating derivatives by central finite differences with the second order of approximation in step $h$ in the strike space (see, e.g. \cite{ItkinBook}), \eqref{localVarappr} can be represented in the form
\begin{align} \label{localVarappr1}
v^0_{j,i} + v^1_{j,i} x_{i} &= \dfrac{b_1 x V_x(x_i) + b_{0,j} V(x_i) - c_j}{V_{x,x}(x_i)}, \\
V_{x}(x_i) &= \alpha_{-1} V(x_{i-1}) + \alpha_{0} V(x_{i}) + \alpha_{1} V(x_{i+1}), , \nonumber \\
V_{x,x}(x_i) &= \delta_{-1} V(x_{i-1}) + \delta_{0} V(x_{i}) + \delta_{1} V(x_{i+1}), \nonumber \\
\alpha_{-1} &= - \dfrac{h_{i+1}}{h_{i}(h_{i+1} + h_{i})}, \quad
\alpha_{0} = \dfrac{h_{i+1}-h_{i}}{h_{i+1}h_{i}}, \quad
\alpha_{1} = \dfrac{h_{i}}{h_{i+1}(h_{i+1} + h_{i})}. \nonumber \\
\delta_{-1} &=  \dfrac{2}{h_{i}(h_{i+1} + h_{i})}, \quad
\delta_{0} = -\dfrac{2}{h_{i+1}h_{i}}, \quad
\delta_{1} = \dfrac{2}{h_{i+1}(h_{i+1} + h_{i})}. \nonumber \\
h_i &= x_i-x_{i-1}, \quad i \in [1,n_j]. \nonumber
\end{align}
Further, associating Put prices $P(S,T_j,x_{i})$ with the given market quotes, the right hands side of the first line in \eqref{localVarappr1} can be found explicitly. This then can be combined with the last line of \eqref{finSystem} to produce a system of $2(n_j-1)$ equations for $v^1_{j,i}$ and $v^1_{j,i}, \ i \in [1,n_j]$. Finally, we take into account the asymptotic behavior of the volatility surface in $x$ at zero and infinity, which, according to our construction, is determined by $v^1_{j,n_j}$ and $v^1_{j,0}$ and is assumed to be known. Thus, we obtain a closed system of $2(n_j-1)$ linear equations with a banded matrix which can be easily solved with a linear complexity. This provides an explicit representation of the local variance function over the whole set of intervals in the strike space determined according to our approximation where the continuous derivatives are replace by finite differences.

Note, that at the first and last strike intervals the approximation of the first and second derivatives by central finite differences should be replaced by one-sided approximations, in more detail see \cite{ItkinBook}, chapter 2.

It could also happen that at some strikes this solution (the smart guess) gives rise to a negative local variance. In such a case we do another step which is a kind of smoothing. Namely, we exclude from the initial guess all values where the local variance is negative and using the remaining points create a spline. Then the negative values in the initial guess are replaced by those given by the constructed spline.

The final step utilizes the exact representation \eqref{combPut} of the Put price in the ELVG model. As the variance function is already known from the previous step, this equation contains two yet unknown constants $\CO{i}{j},\CT{i}{j}$. Accordingly, they can be found by solving the system of 2 linear equations represented by the first and third  lines of \eqref{finSystem}. Then, after this last step is complete, all unknown variables are determined, and thus found solution could be used as an educated initial guess for solving \eqref{finSystem} numerically.

\section{Asymptotic solutions} \label{asympt}

In many practical situations either some coefficients $a_2 = v^1_{j,i}$, or both
$b_2 = v^0_{j,i}, \ a_2 = v^1_{j,i}$ in \eqref{Laplace2} are small. Of course, in that case the general solution \eqref{combPut} remains valid. However, in this case when computing the values of Kummer functions numerically, numerical errors significantly grow. This is especially pronounced when computing the integral $I_{12}$. The main point is that either the Kummer function takes a very small value, and then the constants $\CO{i}{j},\CT{i}{j}$ should be big to compensate, or vice versa. Resolution of this requires a high-precision arithmetics, and, which is more important, taking many terms in a series representation of the Kummer functions, which significantly slows down the total performance of the method.

On the other hand, to eliminate these problems we can look for asymptotic solutions of \eqref{Laplace2} taking into account the existence of small parameters from the very beginning. This approach was successfully elaborated on in \cite{ItkinLipton2017}, and below we proceed in a similar spirit.

\subsection{Small $a_2$}

We can build the solution of \eqref{Laplace3} directly using an independent variable $x$ (so not switching to the variable $z$). We represent it as a series on the small parameter $a_2$, i.e.
\begin{equation} \label{ser}
V(x) = \sum_{i=0}^\infty a_2^i V_i(x).
\end{equation}
In the zero-order approximation by plugging \eqref{ser} into \eqref{Laplace3} and neglecting by terms proportional to $a_2 \ll 1$ we obtain the following equation for $V_0(x)$
\begin{equation} \label{asymA1}
-b_2 V_{xx}(x) + b_1 x V_x(x) + b_0 V(x) = c.
\end{equation}
This equation is simpler than \eqref{combPut}. Still, its solution is given by a general formula
\begin{equation*}
V(x) = C_1 y_1(x) + C_2 y_2(x) + I_{12}(x),
\end{equation*}
\noindent but the fundamental solutions $y_1(x), y_2(x)$ now read
\begin{equation*}
y_1(x) = {\mathcal H} \left(-\frac{b_0}{b_1}, \sqrt{\frac{b_1}{2 b_2}} x\right), \qquad
y_2(x) = M\left(\frac{b_0}{2 b_1}, \frac{1}{2}, \frac{b_1}{2 b_2}x^2\right), \end{equation*}
\noindent where ${\mathcal H}(a, x), \ a,x \in \mathbb{R}$ is the generalized Hermite polynomial $H_a(x)$, \cite{as64}.

\subsection{Small $|z|$}

Based on the definition of $z = (b_2 + a_2 x)b_1/a_2^2$, this could occur in two cases: either at some finite interval in the strike space $|a_2| \gg |b_1 x|, \ |a_2| \gg |b_2|$, or just $z$ is small, so $b_2$ and $a_2$ have the opposite signs.  In any case we have a small parameter under the high-order derivative. This equation belongs to the class of singularly perturbed differential equations, \cite{Wasow1987}. It can be solved by using either the method of matching asymptotic expansions, \cite{Nayfeh2000}, or the method of boundary functions, \cite{VasBut1995}. The latter was used in \cite{ItkinLipton2017} in a similar situation, so for further details we refer a reader to that paper.

However, we can partly eliminate this by constructing solutions of \eqref{Laplace2} using the original variable $x$. Then we have to consider various cases where instead of a small parameter $z$ some other combinations of parameters could be small or large. But if so, a general solution as a function of the original independent variable $x$ could be represented as regular series on the new small parameter. Then, truncating the series, one gets a simplified solution.

To make it more transparent let us represent the general solution of \eqref{Laplace2} expressed in variable $x$, rather than in $z$, as this was done in \eqref{solInhom}
\begin{align} \label{combPutX}
V(x) &= \CO{i}{j} y_1(x) + \CT{i}{j} y_2 (x) + I_{12}(x), \\
y_i(x) &= a_2^k (b_2+a_2 x)^k {\mathcal V}_i\left(
-1 - \frac{b_0}{b_1} + \frac{b_1 b_2}{a_2^2}, 2 - \frac{b_1 b_2}{a_2^2},
\frac{b_1}{a_2^2}(b_2 + a_2 x)\right), \quad i=1,2.\nonumber \\
k &= 1 - \frac{b_1 b_2}{a_2^2}. \nonumber
\end{align}
Observe, that based on the definition of $b_1$ in \cref{finODE}, $b_1 \approx (r-q)\Delta T$, so usually small. Therefore, small $z$ doesn't mean that
$w$ is necessarily small. Below we consider two cases.

\subsubsection{$w \ll 1$}

As $|z| \ll 1$ and $w \ll 1$ we have $w \ll |a_2^2/b_1|$. So $a_2 \ge \sqrt{b_1}$.
In this case $w \ll 1$ is an actual small argument. Therefore, the general solution \eqref{combPutX} can be expanded into series on small $w$. The condition $0 < w \ll 1$ implies that $a_2$ and $b_2$ have the opposite signs.  If $a_2 > 0$ (and so $b_2 < 0$), then in the zero-order approximation we obtain
\begin{align} \label{asymZ}
y_1(w) &= (a_2 w)^{k-1} \left[\dfrac{\Gamma(-k)}{\Gamma(b_0/b_1)} a_2 w + O(w^2) \right] - \left(\frac{b_1}{a_2^3}\right)^{1-k} \frac{\Gamma(k-1)}{b_1\Gamma(k + b_0/b1)}(a_2 b_1 b_2 - b_0 a_2 w) + O(w^2) , \nonumber \\
y_2(w) &= (a_2 w)^{k-1} \left[a_2 w + O(w^2)\right].
\end{align}
As $b_1 > 0$ we have $k-1 > 0$.

If $a_2 < 0$ and $b_2 > 0$, then both RHS in \eqref{asymZ} should be multiplied by a factor
$\exp(-2 i \pi  b_1 b_2/a_2^2)$.

\subsubsection{$a^2_2 \gg |b_1 w|$}

In this case we can also expand the solution in \eqref{combPutX} into series on small $z$ to obtain
\begin{align} \label{asym1}
y_1(z) &= \frac{1}{\Gamma(1 + q_1 - q_2)}\left[\Gamma(1 - q_2) - q_1 \Gamma(- q_2) z\right] + z^{-q_2} \left[\frac{\Gamma(q_2-1)}{\Gamma(q_1)}z + O(z^2)\right] + O(z^2), \\
y_2(w) &= 1 + \frac{q_1}{q_2}z + O(z^2). \nonumber
\end{align}
Note, that based on the definition $q_2 = b_2 b_1/a_2^2$, at large $a_2$ the coefficient $q_2$ could also be small. But $z/q_2 = 1 + a_2 x/b_2 = w/b_2 = O(1)$.

\section{Numerical experiments} \label{numExp}

In our numerical test we use the same data set as in \cite{ItkinSigmoid2015, ItkinLipton2017}. This is done first, to compare performance and a quality of the fit for all those models. Also, we already know that these smiles are difficult to fit precisely, see discussions in \cite{ItkinSigmoid2015, ItkinLipton2017}.

To remind, we take data from \url{http://www.optionseducation.org} on XLF traded at NYSEArca on March 25, 2014. The spot price of the index is $S = 22.64$, and $r = 0.0148,\  q=0.01$. The option implied volatilities ($I_{call}, I_{put}$) are given in Tables~\ref{TabOptC},\ref{TabOptP}. We take all OTM quotes and some ITM quotes which are
very close to the at-the-money (ATM). When strikes for Calls and Puts coincide, we take an average of $I_{call}$ and $I_{put}$ with weights proportional to $1 - |\Delta|_c$ and $1-|\Delta|_p$ respectively, where $\Delta_c, \Delta_p$ are option Call and Put deltas~\footnote{By doing so we do take into account effects reported in \cite{callPutIV}, who pointed out that the IVs calculated from  Call and  Put option prices corresponding to the same strike  do not coincide,  although  they  should  be  equal  in theory. Our weights are chosen according to a pure empirical rule of thumb, and a more detailed investigation of this effect is required.}.
\begin{sidewaystable*}[!htb]
\begin{center}
\footnotesize

\caption{XLF implied volatilities for the Put options.}
\label{TabOptC}

\begin{tabular}{|r|r|r|r|r|r|r|r|r|r|r|r|r|r|r|r|r|r|}
\hline
    \multicolumn{1}{|c|}{\multirow{2}[4]{*}{T}} & \multicolumn{17}{c|}{K,Put} \\
\cline{2-18}          & 10    & 11    & 12    & 13    & 14    & 15    & 16    & 17    & 18    & 19    & 19.5  & 20    & 20.5  & 21    & 21.5  & 22    & 23 \\
\hline
    4/4/2014 &       &       &       &       &       &       &       &       &       &       &       & 39.53 &       & 23.77 & 19.73 & 16.67 &  \\
\hline
    4/11/2014 &       &       &       &       &       &       &       &       &       &       & 35.89 & 30.33 & 26.62 & 22.06 & 18.49 & 16.11 &  \\
\hline
    4/19/2014 &       &       &       &       &       &       &       &       &       & 32.90 &       & 26.79 &       & 20.14 &       & 15.19 & 12.93 \\
\hline
    5/17/2014 &       &       &       &       &       &       &       & 37.66 & 33.27 & 26.88 &       & 23.08 &       & 18.94 &       & 16.12 & 13.86 \\
\hline
    6/21/2014 &       &       &       &       &       & 40.51 & 37.21 & 31.41 & 27.84 & 23.90 &       & 21.07 &       & 18.88 &       & 16.95 & 15.82 \\
\hline
    7/19/2014 &       &       &       &       &       & 36.71 & 33.35 & 29.96 & 26.09 & 22.81 &       & 20.29 &       & 18.13 &       & 16.30 & 14.93 \\
\hline
    12/20/2014 &       &       &       &       & 31.98 & 29.38 & 27.21 & 25.30 & 23.75 & 22.09 &       & 20.67 &       & 19.44 &       & 18.36 & 17.60 \\
\hline
    1/17/2015 & 42.75 & 38.79 & 35.60 & 33.26 & 30.94 & 28.82 & 26.52 & 24.96 & 23.12 & 21.67 &       & 20.29 &       & 19.10 &       & 17.90 & 18.07 \\
\hline
\end{tabular}%

\bigskip
\bigskip

\caption{XLF implied volatilities for the Put options.}
\label{TabOptP}

\begin{tabular}{|r|r|r|r|r|r|r|r|r|r|r|r|r|r|}
\hline
    \multicolumn{1}{|c|}{\multirow{2}[4]{*}{T}} & \multicolumn{13}{c|}{K,Call} \\
\cline{2-14}          & 21    & 21.5  & 22    & 22.5  & 23    & 23.5  & 24    & 25    & 26    & 27    & 28    & 29    & 30 \\
\hline
    4/4/2014 &       & 16.60 & 14.69 & 14.40 & 14.86 &       &       &       &       &       &       &       &  \\
\hline
    4/11/2014 &       & 16.89 & 14.96 & 14.52 & 14.77 & 14.98 &       &       &       &       &       &       &  \\
\hline
    4/19/2014 &       &       & 15.79 &       & 13.38 &       & 15.39 &       &       &       &       &       &  \\
\hline
    5/17/2014 & 16.71 &       & 14.48 &       &       &       & 13.75 &       &       &       &       &       &  \\
\hline
    6/21/2014 & 16.31 &       & 14.78 &       &       &       & 13.92 & 14.28 & 16.58 &       &       &       &  \\
\hline
    7/19/2014 & 16.82 &       & 15.24 &       &       &       & 14.36 & 14.19 & 15.20 &       &       &       &  \\
\hline
    12/20/2014 & 17.63 &       & 16.61 &       &       &       & 15.86 & 15.47 & 15.12 & 15.18 & 15.03 &       &  \\
\hline
    1/17/2015 &       &       &       &       & 16.95 &       & 17.25 & 16.23 & 15.73 & 15.50 & 15.58 & 15.86 & 16.47 \\
\hline
\end{tabular}%

\end{center}
\end{sidewaystable*}

We have already mentioned that in our model for each term the slopes of the smile at plus and minus infinity, $v^1_{j,n_j}$ and $v^1_{j,0}$, are free parameters. So often traders have an intuition about these values. However, in our numerical experiments we take for them just some plausible values. In more detail, for a normalized variance $v(x)$ defined in \eqref{finODE}, for all smiles we use $v^1_{j,0} = -0.1$, and $v^1_{j,n_j} = 0.1$. Accordingly, for the instantaneous variance $\sigma^2(x) = 2 S^2 v(x)/p_j$ the slopes at both zero and plus infinity are time-dependent and can be computed by using the above formula.

When calibrating the model to market data, we use the standard Matlab {\it fsolve} function, and utilize a "trust-region-dogleg" algorithm (see Matlab documentation on {\it fsolve}). Parameter "TypicalX" has to be chosen carefully to speedup calculations.

The results of this calibration which is done term-by term, are given in Fig.~\ref{FigPut}. Here each subplot corresponds to a single maturity $T$ (marked in the legend) and shows market data (discrete points) and computed values (solid line). It can be seen that this simple local calibration algorithm provides a very accurate fit for all terms\footnote{Note, that in \cite{ItkinLipton2017} in the last subplot the fit is not perfect in the vicinity of $X = -0.5$, where $X = \log K/F$ and $F = S e^{(r-q)T}$.}.
\begin{figure}[!htb]
\begin{center}
\fbox{\includegraphics[width = 0.8\textwidth]{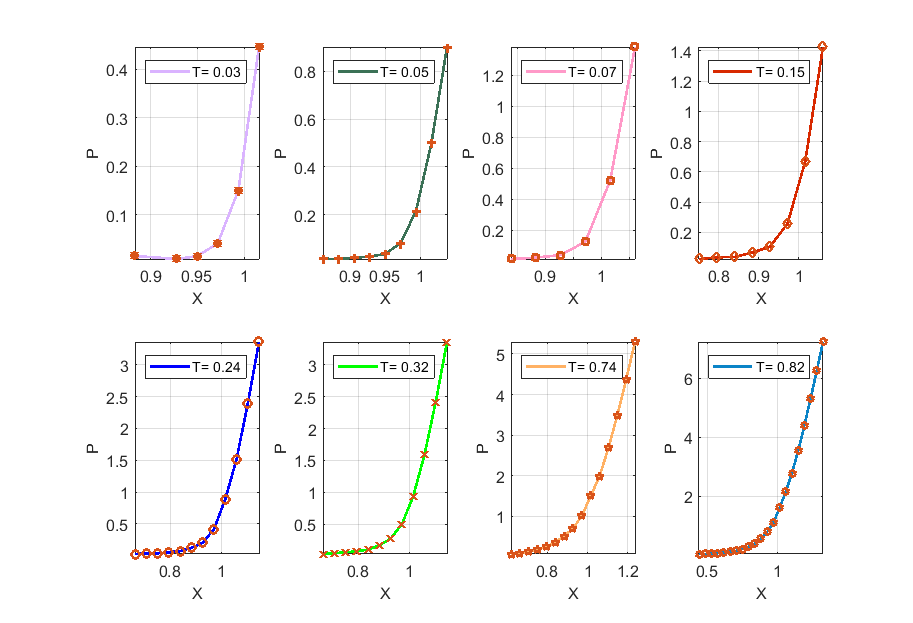}}
\caption{Term-by-term fitting of market Put prices constructed using the whole set of data in Tab.~\ref{TabOptC},\ref{TabOptP}.}
\label{FigPut}
\end{center}
\end{figure}

We constructed the calibration algorithm to be smart enough in a sense that based on the values of parameters at each iteration it decides itself which particular solution (full or asymptotic) should be used at this iteration. We also observed that all full and asymptotic solutions are utilized by the algorithm when calibrating these market smiles.

Table~\ref{Perf} presents some performance measures of our algorithm. It can be seen that the elapsed time depends on the number of iterations and function evaluations necessary to converge to the given tolerance (we use a relative tolerance $\varepsilon = 10^{-4}$). This, in turn, depends on the number of evaluated Kummer functions (for the full solution), or number of exponential and Gamma functions (for the asymptotic solutions). Of course, the asymptotic solutions are much faster to evaluate, therefore an average time to calibrate a typical term is less than a second. For the last term 8 in Tab.~\ref{Perf} calibration is slow for two reasons: i) full solution is used based on the values of parameters, and 2) the number of strikes is higher than for the other terms. But the main reason is that the market data for this term is quite irregular. In any case, performance of this model is much better than that reported in both \cite{ItkinSigmoid2015} and \cite{ItkinLipton2017}.
\begin{table}[H]
\begin{center}
\small
\begin{tabular}{|r|r|r|r|r|r|}
\hline
Term & $T$, years & Elapsed time, sec & iterations & function evaluations & strikes \\
\hline
1 & 0.0274 & 0.86 & 97 & 1202 & 6 \\
\hline
2 & 0.0466 & 2.83 & 97 & 1808 & 9 \\
\hline
3 & 0.0685 & 1.43 & 95 & 1200 & 6 \\
\hline
4 & 0.1452 & 0.64 & 48 & 433 & 8 \\
\hline
5 & 0.2411 & 0.90 & 37 & 470 & 12 \\
\hline
6 & 0.3178 & 2.98 & 82 & 1523 & 12 \\
\hline
7 & 0.7397 & 6.60 & 106 & 3017 & 15 \\
\hline
8 & 0.8164 & 149.67 & 56 & 1317 & 21 \\
\hline
\end{tabular}
\caption{Performance characteristics of the algorithm in the described experiment.}
\label{Perf}
\end{center}
\end{table}

The local variance curves obtained as a result of this fitting are given term-by-term in Fig.~\ref{termXLF}. The corresponding local variance surface is represented in Fig.~\ref{lvXLF}

\begin{figure}[!htb]
\begin{center}
\fbox{\includegraphics[width=0.7\textwidth]{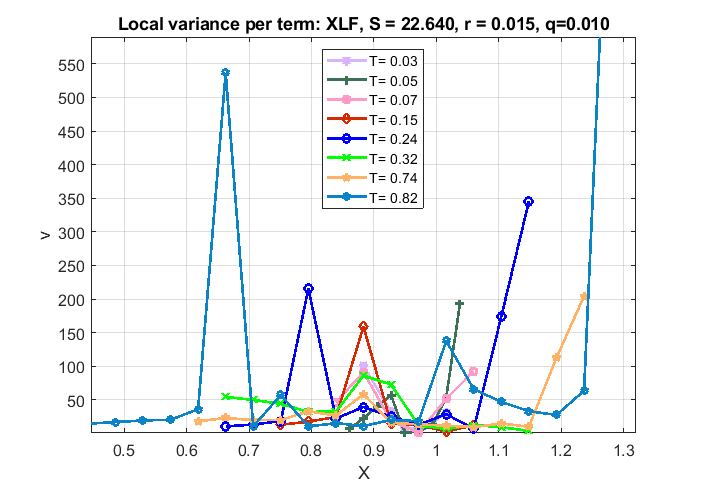}}
\caption{Term-by-term fitting of the instantaneous local variance $\sigma^2(x,T)$.}
\label{termXLF}
\end{center}
\end{figure}

\begin{figure}[!htb]
\begin{center}
\fbox{\includegraphics[width=0.7\textwidth, height=3in]{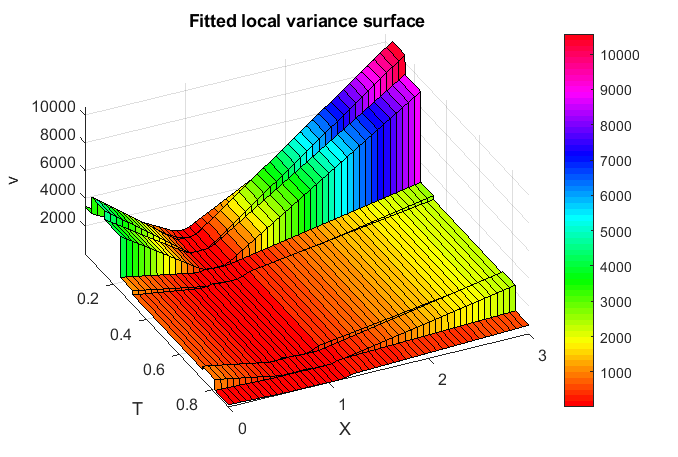}}
\caption{The instantaneous local variance surface $\sigma^2(x,T)$ constructed by using the proposed approach.}
\label{lvXLF}
\end{center}
\end{figure}

By comparing the surface with that given in \cite{ItkinLipton2017}, one can notice that the shape is quite different while for calibration we use the same market smiles. This is because  in \cite{ItkinLipton2017} the standard local volatility model is used, where the underlying price follows a Geometric Brownian motion equipped with an instantaneous local volatility function, while in this paper the model is quite different.

To look at a more regular surface, we proceed with another example which is taken from \cite{Balaraman2016}. In that paper an implied volatility surface of S\&P500 is presented, and the local volatility surface is constructed using the Dupire formula. In our test we take data for the first 12 maturities and all strikes as they are given in \cite{Balaraman2016}, and apply our model to calibrate the local variance surface as this is described in above. When doing so we set $v^1_{j,0} = -0.3$, and $v^1_{j,n_j} = 0.1$ for all smiles.

The results of this calibration are presented in Fig.~\ref{fitSP},\ref{termSP},\ref{lvSP}. By construction, our surface preserves no-arbitrage, while for the approach in \cite{Balaraman2016} they have to solve some additional problems\footnote{As this is mentioned in \cite{Balaraman2016}, the correct pricing of local volatility surface requires an arbitrage free implied volatility surface. If the input implied volatility surface is not arbitrage free, this can lead to negative transition probabilities and/or negative local volatilities and can give rise to mispricing.}.

\begin{figure}[!htb]
\begin{center}
\fbox{\includegraphics[width = 0.8\textwidth]{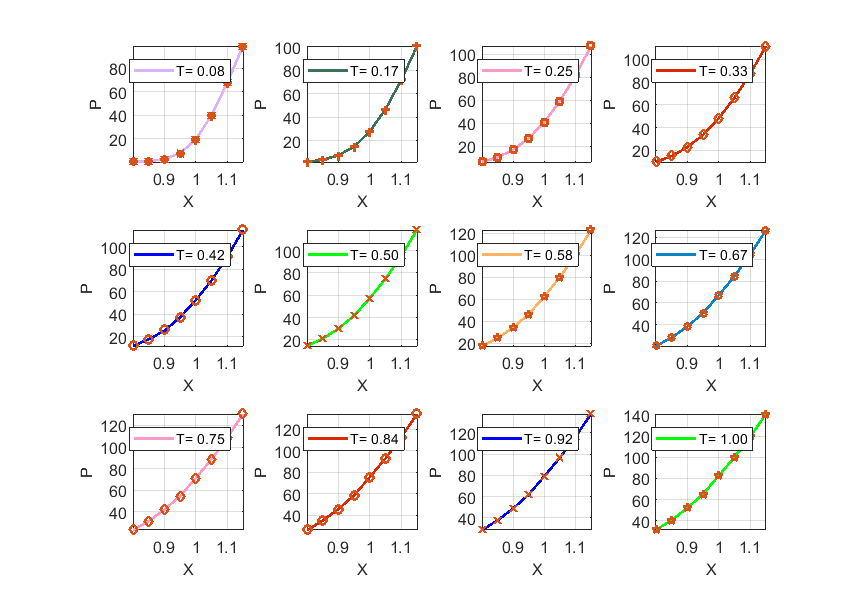}}
\caption{Term-by-term fitting of market S\&P500 Put prices constructed using data of \cite{Balaraman2016}.}
\label{fitSP}
\end{center}
\end{figure}

\begin{figure}[!htb]
\begin{center}
\fbox{\includegraphics[width=0.7\textwidth]{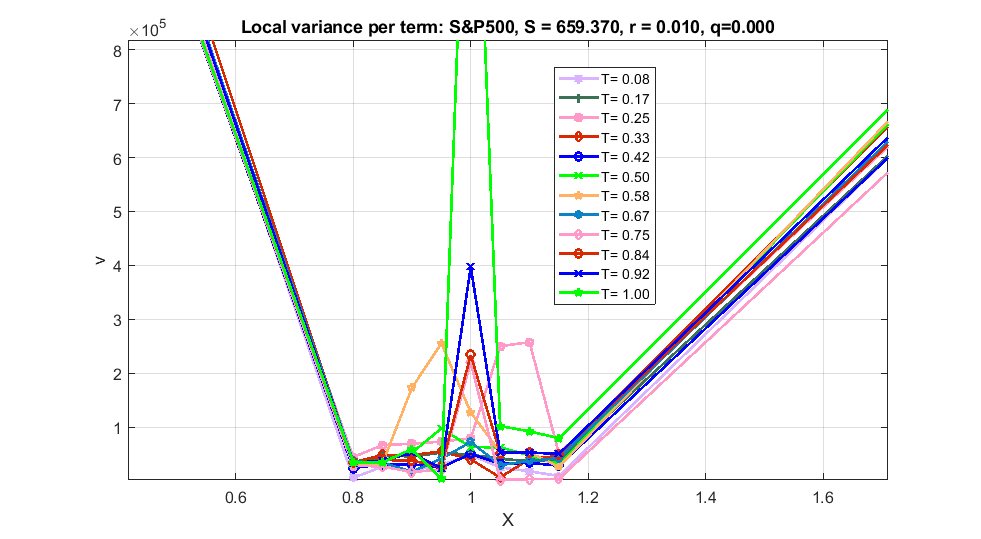}}
\caption{Term-by-term fitting of the instantaneous local variance $\sigma^2(x,T)$ for S\&P500.}
\label{termSP}
\end{center}
\end{figure}

\begin{figure}[!htb]
\begin{center}
\fbox{\includegraphics[width=0.7\textwidth, height=3in]{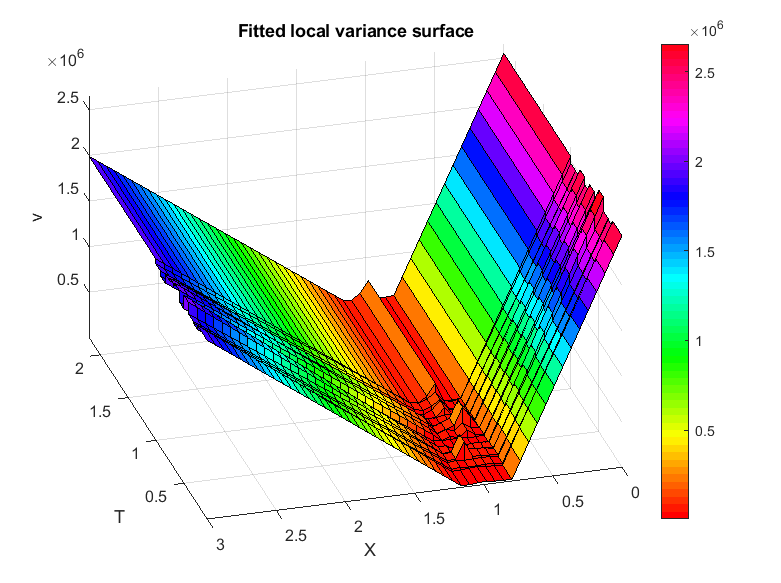}}
\caption{The instantaneous local variance surface $\sigma^2(x,T)$ for S\&P500 constructed by using the proposed approach.}
\label{lvSP}
\end{center}
\end{figure}

In Table~\ref{PerfSP500} we present the performance of our algorithm in this experiment.
It can be seen that here the elapsed time is similar or shorter as compared with the previous test presented in Table~\ref{Perf}.

\begin{table}[H]
\begin{center}
\small
\begin{tabular}{|r|r|r|r|r|r|}
\hline
Term & $T$, years & Elapsed time, sec & iterations & function evaluations & strikes \\
\hline
    1     & 0.0822 & 1.09  & 99    & 1604  & 8 \\ \hline
    2     & 0.1671 & 0.56  & 40    & 377   & 8 \\ \hline
    3     & 0.2521 & 2.32  & 94    & 1615  & 8 \\ \hline
    4     & 0.3315 & 1.70  & 97    & 1186  & 8 \\ \hline
    5     & 0.4164 & 0.10  & 15    & 64    & 8 \\ \hline
    6     & 0.4986 & 2.35  & 111   & 1600  & 8 \\ \hline
    7     & 0.5836 & 2.40  & 111   & 1584  & 8 \\ \hline
    8     & 0.6658 & 2.25  & 131   & 1604  & 8 \\ \hline
    9     & 0.7507 & 1.51  & 95    & 1072  & 8 \\ \hline
    10    & 0.8356 & 2.30  & 98    & 1603  & 8 \\ \hline
    11    & 0.9178 & 0.07  & 13    & 46    & 8 \\ \hline
    12    & 1.0027 & 72.80 & 74    & 795   & 8 \\ \hline
\end{tabular}
\caption{Performance characteristics of the algorithm for calibration of a S\&P500 surface.}
\label{PerfSP500}
\end{center}
\end{table}

\clearpage
\section{Conclusions}

In this paper we propose an expanded version of the Local Variance Gamma model of
\cite{CarrNadtochiy2017} which we refer as an Expanded Local Variance Gamma model, or ELVG. Two main improvements are introduced as compared with the LVG model. First, we add drift to the governing underlying process. It turns out that this a relatively minor (at the first glance) improvement requires a interesting trick to preserve tractability of the model, which is a non-trivial time-change. We show that still in this new model it is possible to derive an ordinary differential equation for the option price which plays a role of Dupire's equation for the standard local volatility model.

The second novelty of the paper as compared with the LVG model is that in contrast to \cite{CarrNadtochiy2017} we consider a local variance to be a piecewise linear function of strike, while in \cite{CarrNadtochiy2017} it was piecewise constant. We proceed in the spirit of \cite{ItkinLipton2017} by describing a no-arbitrage interpolation, and then construct a closed-form solution of our ODE in terms of hypergeometric and generalized hypergeometric functions. An important advantage of this approach is that calibration of the model to market smiles does not require solving any optimization problem, and can be done term-by-term by solving a system of non-linear algebraic equations for each maturity, which, in general, is significantly faster, especially since we provide an algorithm for constructing a smart initial guess. We also provide various asymptotic solutions which allow a significant acceleration of the numerical solution and improvement of its accuracy in the corresponding cases (i.e, when parameters of the model at some iteration obey the conditions to apply the corresponding asymptotic).

In principle, somebody could claim that solving a system of nonlinear equations with a generic solver is not much different from solving a nonlinear optimization problem. Obviously, when our ODE is used as an alternative to the Dupire equation, the difference comes from the fact that calibration based on the Dupire equation requires solving this PDE at every iteration by either numerically, or semi-analytically by using a Laplace transform, which is obviously slower. As was mentioned in Introduction there exist many other calibration algorithms which reduce to a nonlinear optimization problem (e.g., taking a sufficiently large parametric family of local volatility functions and choosing the parameters that provide the best fit of observed prices). For the latter computation of the objective function is fast, but optimization must be constrained to preserve no-arbitrage, and, thus, slow.

In our numerical test we use same market data as in \cite{ItkinSigmoid2015, ItkinLipton2017}. The results of the test demonstrate robustness of the proposed approach from both the speed and accuracy point of view, especially in cases where the above referred papers experienced some difficulties with achieving a perfect fit. An additional test performed for the S\&P500 data taken from \cite{Balaraman2016} gives rise to the same conclusion.

\clearpage

\section*{References}

\newcommand{\noopsort}[1]{} \newcommand{\printfirst}[2]{#1}
  \newcommand{\singleletter}[1]{#1} \newcommand{\switchargs}[2]{#2#1}

\appendix

\end{document}